\documentclass[12pt]{article}
\usepackage{graphicx, amssymb, latexsym, amsfonts, amsmath, lscape, amscd,
amsthm, color, epsfig, mathrsfs, tikz, enumerate, multirow}
\usepackage{hyperref}
\usepackage{cleveref}
\usepackage{bm}
\usepackage{todonotes}
\usetikzlibrary{shapes.misc}
\tikzstyle{noeud}=[circle,inner sep=2, minimum size =3 pt, line width = 1pt, draw=black, fill=white]
\usepackage{subcaption}
\usepackage[utf8]{inputenc}
\usepackage{amsthm}
\usepackage{amsmath}
\usepackage{amsfonts}
\usepackage{amssymb}
\usepackage{graphicx}
\usepackage{epstopdf}
\usepackage{color}
\usepackage{multirow}
\usepackage{mathtools}
\usepackage{xcolor}
\usepackage{a4wide}
\usepackage{bm}
\usepackage{caption}
\usepackage{subcaption}
\usepackage[
  separate-uncertainty = true,
  multi-part-units = repeat
]{siunitx}
\usepackage{rotating}
\usepackage{verbatim}
\usepackage{lineno,hyperref}
\modulolinenumbers[5]
\usepackage{tikz}
\tikzstyle{noeud} = [circle, draw, fill=white, inner sep=2pt]
\newtheorem{theorem}{Theorem}[section]
\newtheorem{lemma}{Lemma}[section]

\newtheorem{observation}{Observation}[section]

\usepackage[noend,linesnumbered,ruled,lined]{algorithm2e}
\SetAlFnt{\scriptsize}
\usepackage{amsfonts}
\usepackage{makecell}

\usepackage[usenames,dvipsnames]{pstricks}
\usepackage{pstricks-add}
\usepackage{epsfig}
\usepackage{pst-grad} 
\usepackage{pst-plot} 
\usepackage[space]{grffile} 
\usepackage{etoolbox} 
\makeatletter 
\patchcmd\Gread@eps{\@inputcheck#1 }{\@inputcheck"#1"\relax}{}{}
\makeatother
\usepackage{mathtools}

\usepackage{graphicx}
\usepackage{tabularx}
\usepackage{textcomp}
\usepackage{xcolor}
\usepackage[labelformat=simple]{subcaption}

\usepackage{booktabs}

\usepackage{pgf,tikz,pgfplots}
\pgfplotsset{compat=1.15}
\usepackage{mathrsfs}
\usepackage{hyperref}
\usepackage{multicol}

\newtheorem{question}{Question}[section]
\newtheorem{definition}{Definition}[section]
\newtheorem{remark}{Remark}[section]
\newtheorem*{answer}{Answer}

\title{\vspace{-1.5cm}Path Connected Dynamic Graphs with a Study of Dispersion and Exploration\footnote{A preliminary version of this work \cite{saxena_2025} is accepted in ICDCN 2025.}}
\author{
{\sc Ashish Saxena}$\,^{(a)}$, {\sc Kaushik Mondal}$\,^{(a)}$\\
\mbox{}\\
{\small $(a)$ Indian Institute of Technology Ropar, India.}\\}

\date{}

\begin{document}
\usetikzlibrary{arrows}
\maketitle
\noindent\textbf{Abstract}

\noindent In dynamic graphs, several edges may get added or deleted in a round. There are different connectivity models based on the constraints on the addition/deletion of edges. One such model is the $T$-Interval Connectivity model, where edges can be added/deleted, keeping the graph nodes connected in each synchronous round. The parameter $T$ depends on the stability of the underlying connected structure across rounds.
There is another connectivity model, namely the Connectivity Time model, where the union of all the edges present in any $T$ consecutive rounds must form a connected graph. This is much weaker than the $T$-Interval Connectivity as the graph may even be disconnected at each round. We, in this work, come up with a new connectivity model, namely $T$-Path Connectivity. In our model, the nodes may not remain connected in each round, but for any pair of nodes $u,v$, there must exist path(s) at least once in any consecutive $T$ rounds. Our model is weaker than $T$-Interval Connectivity but stronger than the Connectivity Time model.

We study the dispersion problem in our connectivity model. Dispersion is already studied in the $1$-Interval Connectivity model. We show that the existing algorithm in $1$-Interval Connected graphs for dispersion with termination does not work in our $T$-Path Connectivity model for obvious reasons. We answer what are the necessary assumptions to solve dispersion in our connectivity model. Then, we provide an algorithm that runs in optimal time with those minimal model assumptions on $T$-Path Connected graphs. Also, we show that solving dispersion is impossible in the Connectivity Time model, even in the presence of several other strong model assumptions. Further, we initiate the study of the exploration problem on these three connectivity models. We provide several impossibility results with different assumptions. In most cases, we establish necessary and sufficient conditions to solve the exploration problem using an optimal number of agents in an asymptomatically optimal time. It is also evident from the studies of dispersion as well as exploration on all the three connectivity models that, Connectivity Time model is indeed the weakest model among these three models.


\section{Introduction}\label{sec:intro}

The graph exploration problem, introduced by Shannon \texorpdfstring{\cite{shannon1993}}{[Shannon, 1993]}, is a fundamental issue in theoretical computer science, particularly in the field of distributed computing involving mobile entities. This problem requires that each node in the graph be visited by one or more mobile computational entities, called agents. Visits can occur a finite number of times (known as exploration with termination) or infinitely often (termed perpetual exploration). In addition to its theoretical significance, the exploration problem is practically relevant in networks that support mobile entities, such as software agents, vehicles, or robots. By visiting all nodes, agents can identify nodes that may have issues within the network, disseminate data throughout the system, or gather specific information from the entire network.

The dispersion problem, introduced by Augustine \cite{Augustine_2018}, involves the coordination of $n$ agents on an $n$ node graph to reach a configuration where one agent is present on each node. This problem is generally applicable to real-world scenarios in which $n$ agents must coordinate and share $n$ resources located at various places. The goal is to minimize the total cost of solving the problem in such applications where the cost of agents moving around on the graph is much lower than the cost of sharing a resource by multiple agents. The dispersion problem is also connected to the exploration problem \cite{Pelc_2004}. Any solution to the dispersion of $n$ agents can be applied as a solution to the exploration with $n$ agents as long as the assumptions and model parameters remain the same. Therefore, the exploration is closely related to dispersion.

The dynamic nature of modern networks introduces new challenges in solving different algorithmic problems in the field of mobile computing and beyond as these networks keep changing over time. From the perspective of mobile agents, agents may need to carry out their tasks while their surroundings evolve. Such a dynamic graph model was introduced by Kuhn et al. \cite{Kuhn_2010} in 2010. In this dynamic graph model, Kuhn et al. \cite{Kuhn_2010} introduce a stability property called $T$-Interval Connectivity (for $T \geq 1$), which stipulates that for each $T$ consecutive rounds, there exists a stable, connected spanning sub-graph, while several other edges may get added or deleted in each and every round. The formal description of their model is as follows. Let $V$ be a set of static vertices, $S=\{(u,\,v)\,|\, u,v\in V\}$, where $(u,\,v)$ denotes an edge between $u$ $\&\;v$, and $\mathscr{P}(S)$ be the power set of the set $S$. A synchronous dynamic network is modeled as a dynamic graph $G = (V, \,E)$, where $V$ is a static set of nodes and $E : \mathbb{N } \rightarrow \mathscr{P}(S)$ is a function that maps a round number $r \in \mathbb{N} \cup \{0\}$ to a set $E(r) \in \mathscr{P}(S)$ of undirected edges. For any round $r \geq 0$, we denote the graph by $\mathcal{G}_r=(V, \, E(r))$. Kuhn et al. \cite{Kuhn_2010} provide the following definition of connectivity in the network.


\begin{definition}
    \cite{Kuhn_2010} ($T$-Interval Connectivity) A dynamic graph $G = (V, \,E)$ is $T$-Interval Connected for $T \geq 1$ if for all $r \in \mathbb{N} \cup \{0\}$, the static graph $G_{r,\,T} := (V,\, \bigcap_{i=r}^{r+T-1} E(i))$ is connected. The graph is said to be $\infty$-Interval Connected if there is a connected static graph $G'=(V, \,E')$ such that for all $r \in \mathbb{N} \cup \{0\}$, $E' \subseteq E(r)$.
\end{definition}
For $T>1$, the graph can not change arbitrarily as it needs to maintain a stable spanning sub-graph. However, for $T = 1$, this means that the graph is connected in every round, but may change arbitrarily between rounds. Therefore, we have the following observation.

\begin{observation}\label{obs:T-Interval}
A dynamic graph $G$ with $T$-Interval Connectivity also holds 1-Interval Connectivity, but the converse may not be true.
\end{observation}

Later in 2014, Michail et al. \cite{Michail_2014} introduced another natural and practical definition of connectivity of a possibly disconnected dynamic network that they call Connectivity Time. The authors provide the following definition of connectivity.

\begin{definition}
    \cite{Michail_2014} (Connectivity Time)  The connectivity time of a dynamic network $G = (V, \,E)$ is the minimum $T \in \mathbb{N}$ s.t. for all times $r \in \mathbb{N} \cup \{0\}$ the static graph $G_{r,\,T}:= (V,\, \bigcup_{i=r}^{r+T-1} E(i))$ is connected.
\end{definition}

It is important to note that in this dynamic graph model, if $T=1$, then it is nothing but 1-Interval Connectivity. The basic difference between the connectivity model of \cite{Kuhn_2010} and \cite{Michail_2014} is as follows. For $T>1$, in $T$-Interval Connected model, $\mathcal{G}_r$ is always connected for each round $r$, and in Connectivity Time model, $\mathcal{G}_r$ may be disconnected for each round $r$. Therefore, we can say a graph which holds $T$-Interval Connectivity also holds the Connectivity Time property, but the converse may not be true. Therefore, $T$-Interval Connectivity \cite{Kuhn_2010} is a stronger connectivity model than the Connectivity Time model \cite{Michail_2014}. In this work, we introduce a new connectivity property which we call $T$-Path Connectivity, and this connectivity property lies between the aforementioned two connectivity properties. The formal description of $T$-Path Connectivity is as follows.

\begin{definition}
($T$-Path Connectivity) A dynamic graph $G = (V, \,E)$ is $T$-Path Connected for $T \geq  1$, if for all $r \in \mathbb{N} \cup \{0\}$ and for any pair of nodes $u$, $v \in V$, there exists at least one round $i \in [r,\, r+T-1] $ such that $u$, $v$ are in the same connected component of $\mathcal{G}_i$.
\end{definition}

\begin{figure}[ht]
    \centering
    \includegraphics[width=0.75\linewidth]{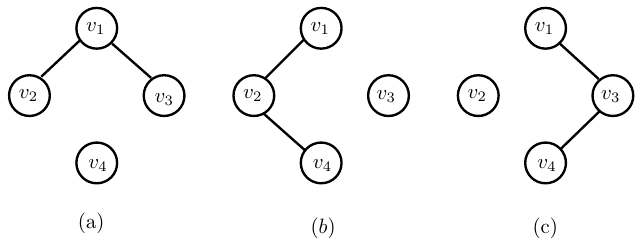}
     \caption{(a) Graph $\mathcal{G}_r$, where $r$(mod 3)=0, (b) Graph $\mathcal{G}_{r}$, where $r$(mod 3)=1, (c) Graph $\mathcal{G}_{r}$, where $r$(mod 3)=2. This figure is an example of $T$-Path Connectivity for $T=3$.}
     \label{fig:exp}
\end{figure}

\begin{figure}[ht]
    \centering
    \includegraphics[width=0.75\linewidth]{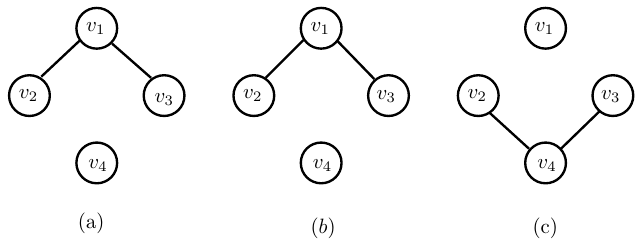}
    \caption{(a) Graph $\mathcal{G}_r$, where $r$(mod 3)=0, (b) Graph $\mathcal{G}_{r}$, where $r$(mod 3)=1, (c) Graph $\mathcal{G}_{r}$, where $r$(mod 3)=2. This figure is an example of Connectivity Time for $T=3$.}
     \label{fig:exp1}   
\end{figure}

 Note that all three connectivity definitions are equivalent in the case of $T=1$. Now we discuss the differences between our connectivity model and the existing connectivity models. For $T>1$, the basic difference between $T$-Interval Connected graph and $T$-Path Connected graph is as follows. In $T$-Interval Connected graph $\mathcal{G}_r$ is connected for all $r \in \mathbb{N} \cup \{0\}$, and in $T$-Path Connected graph $\mathcal{G}_r$ may even be disconnected for all $r \in \mathbb{N} \cup \{0\}$. Following is one such example. Let $V=\{v_1,\,v_2,\,v_3,\,v_4\}$ and $T=3$. At the round $r$, if $r$ (mod 3)=0, $\mathcal{G}_r=(V,\, E(r))$, where $E(r)=\{(v_1, \,v_2), \, (v_1, \, v_3)\})$(refer Fig. \ref{fig:exp}(a)). At the round $r$, if $r$ (mod 3)=1, $\mathcal{G}_r=(V,\, E(r))$, where $E(r)=\{(v_1, \,v_2), \, (v_2, \, v_4)\}$ (refer Fig. \ref{fig:exp}(b)). At the round $r$, if $r$(mod 3)=2, $\mathcal{G}_r=(V,\, E(r))$, where $E(r)=\{(v_1, \,v_3), \, (v_3, \, v_4)\}$ (refer Fig. \ref{fig:exp}(c)). Therefore, $G=<\mathcal{G}_0, \, \mathcal{G}_1, \ldots>$. In this example, we can see that if there is no path between $v_i$ and $v_j$ in some round $r$, then within the next two rounds, there is a path between $v_i$ and $v_j$. Hence $G=<\mathcal{G}_0, \, \mathcal{G}_1, \ldots>$ maintains $T$-Path Connectivity but not $T$-Interval Connectivity.

For $T>1$, the basic difference between the Connectivity Time graph and the $T$-Path Connected graph is as follows. In $T$-Path Connected graph, if there is no path in $\mathcal{G}_r$ at round $r$ between two nodes (say $u$, $v$), then within next consecutive $T$ rounds, there exist $t'$ ($r<t'<r+T$) such that there is at least one path in $\mathcal{G}_{t'}$. In Connectivity Time graphs, it is not necessarily true. We can understand this through the following example. Let $V=\{v_1,\,v_2,\,v_3,\,v_4\}$ and $T=3$. At the round $r$, if $r$ (mod 3)=0 or 1, $\mathcal{G}_r=(V,\, E(r))$, where $E(r)=\{(v_1, \,v_2), \, (v_1, \, v_3)\})$(refer Fig. \ref{fig:exp1}(a) and \ref{fig:exp1}(b)). At the round $r$, if $r$ (mod 3)=2, $\mathcal{G}_r=(V,\, E(r))$, where $E(r)=\{(v_4, \,v_2), \, (v_4, \, v_3)\})$(refer Fig. \ref{fig:exp1}(c)). Since, for any round $r$, $G_{r,\,3}= (V,\, \bigcup_{i=r}^{r+2} E(i))$ is connected, therefore, this example holds the Connectivity Time property but does not hold $T$-Path Connectivity as there is no path between $v_1$ and $v_4$ in any round $r$. We have the following observation based on these models.

 \begin{observation}\label{obs:connectivity}
    Any dynamic graph $G$ with $T$-Interval Connectivity is also a dynamic graph $G$ with $T$-Path Connectivity. And, any dynamic graph $G$ with $T$-Path Connectivity is also a dynamic graph $G$ with Connectivity Time. 
 \end{observation}

 \vspace{0.2cm}
\noindent\textbf{Motivation for $T$-Path Connectivity:} Let us motivate our model through a practical example. Consider a scenario where some application being run over the network that uses global communication. In Figure \ref{fig:exp1}, the adversary ensures that the network maintains the Connectivity Time Property at all the times.  Consider node \( v_1 \), which holds essential data that must be shared across the network. A critical observation is that \( v_1 \) and \( v_4 \) are never part of the same connected component in \( \mathcal{G}_r \) for any round \( r \geq 0 \). Consequently, \( v_1 \) has no direct way to communicate with \( v_4 \) and must rely on intermediate nodes specifically \( v_2 \) or \( v_3 \) as nodes \( v_2, v_3, \) and \( v_4 \) are in a connected component in \( \mathcal{G}_r \) whenever \( r \equiv 2 \pmod{3} \). However, the presence of faults introduces a critical vulnerability. If both \( v_2 \) and \( v_3 \) fall under the adversary’s control, they can deliberately refuse to relay information during the rounds when they are connected to \( v_4 \). As a result, \( v_4 \) remains perpetually deprived of the crucial data from \( v_1 \).  However, our proposed \( T \)-Path Connectivity model gives more flexibility. In $T$-Path Connectivity, within three consecutive rounds, there is at least one round where \( v_1 \) and \( v_4 \) are part of the same connected component, enabling successful information transfer. This flexibility makes \( T \)-Path Connectivity a more resilient and practical model for dynamic networks, particularly in adversarial environments.

In this study, we investigate both the dispersion and exploration problems across three different connectivity models. In the following section, we begin by discussing the current state of dispersion in dynamic graphs, followed by an overview of the exploration problem in the dynamic graphs.

\section{Related Work}
In this section, we examine existing research on the dispersion and exploration problems within dynamic graphs. Although both problems involve agents moving through a changing environment, they have different objectives and constraints. Dispersion emphasizes ensuring that agents occupy distinct nodes, while exploration focuses on visiting all nodes in the graph. We will first discuss previous work on dispersion in dynamic graphs, followed by a review of existing approaches to exploration in these settings.
\subsection{Status of Dispersion on Dynamic Graphs}\label{sec:related}
Dispersion of the 1-Interval Connected dynamic graphs is studied by Agarwalla et al. \cite{Agarwalla_2018} and Kshemkalyani et al. \cite{Ajay_dynamicdisp}. In \cite{Agarwalla_2018}, the authors develop several deterministic algorithms for agents to achieve dispersion on a dynamic ring. In \cite{Ajay_dynamicdisp}, the authors use a weaker model where the adversary can add or remove edges. The authors show that it is impossible to solve dispersion on a $1$-Interval Connected dynamic graph $G$ in the face-to-face communication model, even if $1$-neighbourhood knowledge is available to the agents (i.e., an agent located at a particular node (say $w$) can see the set of nodes that are neighbours of $w$, as well as the edges connecting these nodes to node $w$) and each agent has unlimited memory. Also, it is impossible to solve dispersion on a $1$-Interval Connected dynamic graph in the global communication model without $1$-neighborhood knowledge, even with unlimited memory at each agent. They provided an asymptotically optimal $\Theta(k)$ rounds algorithm in the global communication model with $1$-neighborhood knowledge by each agent. In \cite{Ajay_dynamicdisp}, $\mathcal{G}_r$ is always connected for each round $r$ but $\mathcal{G}_r$ may not be connected for any round $r$ in our model. In their work, agents fill at least one vacant node at each round as they can use the global communication model to the best effect using the connectivity of the graph nodes in each round. In \cite{Ajay_dynamicdisp}, the agents do not need any extra information to solve the dispersion problem with termination (except global communication and 1-neighbourhood knowledge as these two are necessary to solve dispersion) i.e., each agent knows the dispersion has been achieved and terminates (explicit dispersion). It is because, with 1-Interval Connectivity, $\mathcal{G}_r$ is connected. If there is any node with more than one agent, then while communicating with each other using global communication, they can understand this. If, in some rounds, they do not get such information, all agents understand the dispersion has been achieved and terminated.

In \cite{Ajay_dynamicdisp}, the discussion of $T$-Interval Connected graphs is left as future work. For $T>1$, the $T$-Interval Connectivity model is stronger than the $1$-Interval Connectivity model due to the following reasons. For $T>1$, the graph cannot change arbitrarily as it needs to maintain a stable spanning sub-graph for each $T$ consecutive round. However, for $T = 1$, the graph is connected in every round but may change arbitrarily between rounds. Therefore, it is unknown whether global communication and $1$-neighbourhood knowledge are necessary assumptions to solve dispersion in $T$-Interval Connected graphs or not. Since every $T$-Interval Connected dynamic graph is 1-Interval Connected dynamic graph, the algorithm mentioned in \cite{Ajay_dynamicdisp} also solves dispersion in $T$-Interval Connected graphs if agents are equipped with 1-hop visibility and global communication.

To the best of our knowledge, the dispersion problem has not been studied previously in the Connectivity Time model and $T$-Path Connectivity model. In this work, we study the dispersion problem on Connectivity Time graphs and $T$-Path Connected graphs. In \emph{Connectivity Time} model, we will show in Section \ref{sec:existing} that it is impossible to solve the dispersion problem even if agents are equipped with infinite memory, full visibility, global communication, and know the size of the team of agents and number of nodes. Then we study the dispersion problem on the $T$-Path Connectivity model and provide an algorithm in optimal time with minimal model assumptions.

\subsection{Status of Exploration on Dynamic Graphs}\label{sec:related_exp}
Many studies on the exploration of dynamic graphs are centralized (or offline), meaning they assume that exploring agents have complete prior knowledge of the topological changes and their occurrence times. These studies include the analysis of the complexity of computing an optimal exploration schedule under the 1-Interval Connectivity assumption \cite{flocchini2012searching}, which was generalized and extended in \cite{Erlebach_2015} and later in \cite{erlebach2018faster, erlebach2019two}. Additionally, they cover the computation of exploration schedules for rings under the more stringent $T$-Interval Connectivity assumption \cite{ilcinkas2018exploration}, as well as for cactuses under the 1-Interval Connectivity assumption \cite{ilcinkas2014exploration}.

Fewer studies have employed a distributed approach to exploration. On the probabilistic side, there is a seminal work on random walks \cite{avin2008explore}. On the deterministic side, exploration has been studied under specific constraints related to network connectivity and its underlying topology. Research on exploration with termination by a single agent in periodic temporal networks, including carrier networks, has been conducted in \cite{flocchini2012searching, flocchini2013exploration, ilcinkas2011power, ilcinkas2018exploration}. The topic of perpetual exploration by three agents in temporally connected rings has been studied in \cite{bournat2016self,bournat2017computability}. Additionally, the exploration with termination in 1-Interval Connected rings by two and three agents has been addressed in \cite{di2020distributed}. This study considered not only the traditional fully-synchronous (Fsync) scheduler, where all agents are active in each round but also a semi-synchronous (Ssync) scheduler, where only a subset of agents is active in each round. Further research in \cite{gotoh2018group} investigated exploration with termination by $O(n)$ agents in dynamic tori of size $m \times n$ (with $m \leq n$), where each column and row is structured as a 1-Interval Connected ring. The exploration with termination by a single agent, possessing partial information about dynamic changes, has been studied in \cite{gotoh2019exploration} for 1-Interval Connected rings.

The most recent work on the perpetual exploration in the time-varying graphs is given by Gotoh et al.in \cite{GOTOH2021}. The authors focus on the solvability of the exploration of such dynamic graphs, and specifically on the number of agents that are necessary and sufficient for exploration under the Fully synchronous and Semi-synchronous activation schedulers. In the time-varying graph, there is a footprint $G$ from where the adversary deletes edges, and the agents understand these deleted edges when they try to move, and the movement becomes unsuccessful. In our model, there is no footprint i.e., the adversary can delete or add edges arbitrarily while maintaining the connectivity property. More importantly, the port numbers in the time-varying graphs are fixed, while in the dynamic graphs, the port numbers are not fixed due to the unavailability of the footprint, and it depends on the degree of nodes in $\mathcal{G}_r$. Therefore, in our model, movements are always successful, and agents can not sense any missing edge. This makes our model much weaker than the time-varying graph model used in \cite{GOTOH2021}.
To the best of our knowledge, no one has previously investigated the exploration problem in such dynamic graphs. In this work, we study the exploration problem across all three connectivity models and determine the lower bounds regarding time and the number of agents.

\section{Model and Problem Definitions}\label{sec:model}
\noindent \textbf{Dynamic graph model:} Let $G = (V, \,E)$ be a dynamic graph where $V$ is a static set of nodes with $|V | = n$. Let $S=\{(u,\,v)\,|\, u,v\in V\}$, where $(u,\,v)$ denotes an edge between $u$ $\&\;v$, and $\mathscr{P}(S)$ be the power set of the set $S$. The map $E: \mathbb{N } \rightarrow \mathscr{P}(S)$ is a function mapping a round number $r \in \mathbb{N} \cup \{0\}$ to a set of undirected edges $E(r)$. For any round $r \geq 0$, we denote the graph by $\mathcal{G}_r=(V, \, E(r))$. Let $|E(r)| = m_r$ be the number of edges in round $r$. The dynamic network $G$ is given by a sequence of undirected graphs $<\mathcal{G}_0, \, \mathcal{G}_1, \, \mathcal{G}_2, \, \ldots >$. We assume that there is an adversary which can add or remove edges arbitrarily at the beginning of round $r$. We denote the degree of $v \in \mathcal{G}_r$ at round $r$ by $deg_r(v)$. Similarly, the maximum degree of the graph $\mathcal{G}_r$ is the maximum among the degrees $deg_r(v)$ of the nodes in $\mathcal{G}_r$. The diameter $D_r$ of $\mathcal{G}_r$ is the longest shortest path between any two nodes in $\mathcal{G}_r$. The dynamic diameter $\hat{D}$ of $G$ is $max \, D_r$, where $1\leq r< \infty$. The graph $\mathcal{G}_r$ is an unweighted and undirected graph. In addition, $\mathcal{G}_r$ is anonymous, i.e., nodes have no (unique) IDs and hence are indistinguishable from each other. The graph $\mathcal{G}_r$ is a port-labelled graph i.e., the ports of any node $v \in \mathcal{G}_r$ have unique labels in the range $[0,\, deg_r(v)-1]$. Any edge $e(u, \,v)$ connecting two nodes $u,\, v \in \mathcal{G}_r$, has two port numbers associated with it, one at $u$ and one at $v$, and there is no relation between these two port numbers. There is no relation between the port labels of $\mathcal{G}_r$ and $\mathcal{G}_{r'}$ when $r \neq r'$. There is no storage at nodes of $G$. If there is no agent at node $v\in \mathcal{G}_r$, then we call such node a \underline{\emph{hole}} at round $r$. If there are two or more agents at node $v \in \mathcal{G}_r$, then we call such a node a \underline{\emph{multinode}} at round $r$.

\vspace{0.2cm}
\noindent \textbf{Agent model:} We consider $k \leq n$ agents placed arbitrarily on the nodes of the graph $G$. Each agent has a unique identifier assigned from the range $[1, k]$. Each agent knows its ID but has no information of other agents' IDs. The agents are equipped with memory. In any round $r$, the agent can see the degree of the node where it is currently located in the underlying graph $\mathcal{G}_r$ and also the port numbers corresponding to each edge incident to that node. The algorithm runs in synchronous rounds. In each round $t$, an agent $a_i$ performs one \textit{Communicate-Compute-Move} (CCM) cycle as follows:
\begin{itemize}
    \item \textbf{Communicate:} Agent $a_i$ can communicate with other agents $a_j$ (present on the same node $v_i$ of $a_i$ or on a different node) depending on the communication model used. 
    \item \textbf{Compute:} The agent does some computation including computing the port number it will move through at the end of the current round or decides not to move at all.
    \item \textbf{Move:} It moves through the computed port, if any.
\end{itemize}

\noindent Time complexity is measured by the number of rounds starting from the round when all agents become active till the round when the last active agent(s) terminates.

\vspace{0.2cm}
\noindent \textbf{Visibility model:} There are two types of visibility models: zero-hop visibility and  $l$-hop visibility. In the zero-hop visibility model \cite{Augustine_2018, Shibata_2016, Molla_2019, Molla_2020}, an agent located at a particular graph node (say $v \in \mathcal{G}_r$) can see agents present at node $v$ at round $r$, and port numbers associated with node $v$. In $l$-hop visibility \cite{Agarwalla_2018, Avery_2020}, an agent located at a particular node (say $w \in \mathcal{G}_r$) can see the subgraph induced by the set of nodes $S_l$ that are within distance $l$ from $w$ in round $r$. It can also see the presence/absence of agents in all the nodes of this sub-graph. In round $r$, if $l=D_r$, we call it full visibility. We consider 1-hop visibility for our algorithms.

\vspace{0.2cm}
\noindent \textbf{Communication model:} There are two communication models in literature, face-to-face (f-2-f) and global. In the f-2-f communication model \cite{Augustine_2018, Ajay_2019, Molla_2019}, an agent located at a particular graph node is only able to communicate with other agents present at that same node. In contrast, the global communication model \cite{Pelc_2006, Das_2018, Ortolf_2012, Molla_2020, Ajay_2020} allows an agent at a graph node to communicate with any other agent present in the same connected component of the graph, and this type of communication happens through the links of the graph. If there are two different connected components, say $G_1$, $G_2$ of $\mathcal{G}_r$, then there is no communication in round $r$ between the agents present in $G_1$ and the agents present in $G_2$ as per the existing global communication model just because there are no links between $G_1$ and $G_2$ in round $r$. We consider global communication for our algorithm. We study the following two variants of dispersion.

\begin{definition}
    \textbf{(Implicit Dispersion)} 
A set of  $k \leq n$ mobile agents is initially placed arbitrarily on the nodes of a graph $G$ of size $n$. The agents need to reposition themselves such that each node of $G$ contains at most one agent.
\end{definition}

\begin{definition}
    (\textbf{Explicit Dispersion}) 
A set of  $k \leq n$ mobile agents is initially placed arbitrarily on the nodes of a graph $G$ of size $n$. The agents need to reposition themselves such that each node of $G$ contains at most one agent. Further, each agent must terminate when such a configuration is achieved \footnote{It is important to note that an algorithm that solves explicit dispersion also solves implicit dispersion, but the converse may not be true.}.
\end{definition}

We also study the following two variants of the exploration problem.

\begin{definition}
    (\textbf{Exploration})
Exploration by a set of mobile agents is the problem where each node of the underlying graph is visited by at least one mobile agent before the agents terminate. 
\end{definition}

\begin{definition}
    (\textbf{Perpetual Exploration})
Perpetual exploration is the problem where a set of mobile agents continuously visits all the nodes of a given graph infinitely often.
\end{definition}

\begin{table}[ht]
\centering
\caption{Necessary and sufficient assumptions to solve dispersion in different connectivity models considering agents do not know $n$, $k$. In this table, existing work is in italic $\&$ blue colour text form, and our work is in normal text form.}
\tiny
\begin{tabular}{|c|c|c|c|c|c|c|}
 \hline
 \raisebox{0pt}[1em][2em]{\textbf{\makecell{Dynamic\\ graph model}}}& \multicolumn{2}{|c|}{\raisebox{0pt}[2.5em][1.2em]{\textbf{\makecell{Implicit\\ Dispersion}}}} & \multicolumn{2}{|c|}{\raisebox{0pt}[2.5em][1.2em]{\textbf{\makecell{Explicit \\Dispersion}}}} & \multicolumn{2}{|c|}{\raisebox{0pt}[2.5em][1.2em]{{\textbf{Algorithm}}}}\\
\hline
\raisebox{0pt}[2.5em][1.2em]{} & \raisebox{0pt}[2.5em][1.2em]{\textbf{\makecell{Necessary \\Assumptions}}} & \raisebox{0pt}[2.5em][1.2em]{\textbf{\makecell{Sufficient\\ Assumptions}}} & \raisebox{0pt}[2.5em][1.2em]{\textbf{\makecell{Necessary \\Assumptions}}} & \raisebox{0pt}[2.5em][1.2em]{\textbf{\makecell{Sufficient\\ Assumptions}}} & \textbf{Implicit} & \textbf{\makecell{Explicit}}\\
\hline\hline
\raisebox{0pt}[2.5em][2em]{\makecell{ Connectivity\\ Time \\(This Work)}} & \raisebox{0pt}[2.5em][1.2em]{\makecell{Impossible\\ to solve}} & \raisebox{0pt}[2.5em][1.2em]{\makecell{Impossible\\ to solve}} & \raisebox{0pt}[2.5em][1.2em]{\makecell{Impossible\\ to solve}} & \raisebox{0pt}[2.5em][1.2em]{\makecell{Impossible\\ to solve}} & \textbf{} & \textbf{\makecell{}}\\
\hline
\raisebox{0pt}[2.5em][1.2em]{\makecell{\cite{Ajay_dynamicdisp} 1-Interval\\ Connectivity}} & \raisebox{0pt}[2.5em][1.2em]{\textcolor{blue}{\emph{\makecell{Global comm,\\ 1-hop visibility}}}} & \raisebox{0pt}[2.5em][1.2em]{\textcolor{blue}{\emph{\makecell{Global comm,\\ 1-hop visibility}}}} & \raisebox{0pt}[2.5em][1.2em]{\textcolor{blue}{\emph{\makecell{Global comm,\\ 1-hop visibility}}}} & \raisebox{0pt}[2.5em][1.2em]{\textcolor{blue}{\emph{\makecell{Global comm,\\ 1-hop visibility}}}} & \textcolor{blue}{\emph{\makecell{$\Theta(k)$-time, \\ $\Theta(\log k)$-memory\\ per agent }}} & \textcolor{blue}{\emph{\makecell{$\Theta(k)$-time, \\ $\Theta(\log k)$-memory\\ per agent }}}\\
\hline
\raisebox{0pt}[2.5em][2em]{\makecell{$T$-Path \\Connectivity, $T>1$\\ (This Work)}} & \raisebox{0pt}[2.5em][1.2em]{\makecell{Global comm,\\ 1-hop visibility}} & \raisebox{0pt}[2.5em][1.2em]{\makecell{Global comm\\ 1-hop visibility}} & \raisebox{0pt}[2.5em][1.2em]{\makecell{Global comm,\\ 1-hop visibility, \\at least one of $n$, $k$, $T$}} & \raisebox{0pt}[2.5em][1.2em]{\makecell{Global comm,\\ 1-hop visibility,\\ knowledge of $T$}} & \makecell{$\Theta(k\cdot T)$-time, \\ $\Theta(\log k)$-memory\\ per agent } & \makecell{$\Theta(k\cdot T)$-time, \\ $O(\log max(k, \, T)$)-\\memory per agent }\\
\hline
\hline
\end{tabular}
\label{tab:minimal} 
\end{table}
\section{Our Contribution}
In this work, we introduce a new variant of the connectivity model, which we call $T$-Path Connectivity (refer to Section \ref{sec:intro}). We study the dispersion and exploration problems related to different connectivity models. The details are provided below.

\subsection{Dispersion Problem}
In this work, we present impossibility results for the $T$-Path Connectivity and Connectivity Time models. Additionally, we outline the necessary and sufficient conditions for solving the explicit dispersion problem in $T$-Path Connected graphs. The details of our contributions are summarized as follows.

\begin{itemize}

\item We show that if the adversary maintains the Connectivity Time property, then it is impossible to solve the implicit dispersion problem in the Connectivity Time graphs (refer Theorem \ref{thm:imp-connec}). This result is valid even if agents are equipped with infinite memory, full visibility, and global communication, and know the parameters $k$, $n$, $T$.
    \item If agents are equipped with global
communication, full visibility, and do not know $k$, $n$, $T$, then it is impossible
to solve explicit dispersion in $T$-Path Connected graphs (refer Theorem \ref{thm:imp:knowledge_T}). 

    \item We provide a time lower bound $\Omega(k\cdot T)$ to solve implicit dispersion in $T$-Path Connected graphs. This result also holds even if the agents have infinite memory, can do global communication and have full visibility (refer Theorem \ref{thm:lower_bound_T}).
    \item We provide $O(k \cdot T)$-time and $O(\log$ max$(T, \, k))$ memory per agent algorithm to solve explicit dispersion in $T$-Path Connected graphs when agents are equipped with 1-hop visibility, global communication, and knows parameter $T$ (refer Section \ref{sec:algorithm}).
    \item To solve implicit dispersion in $T$-Path Connected graphs, agents do not need the knowledge of $T$ when agents are equipped with 1-hop visibility and global communication (refer Remark 2 of Section \ref{sec:analysis-T-Path}).
\end{itemize}

Refer to Table \ref{tab:minimal} for necessary and sufficient assumptions to solve implicit/explicit dispersion in different dynamic graph models. 

\begin{table}
\centering
\caption{Necessary and sufficient assumptions to solve exploration in different connectivity models considering agents do not know $n$, $k$, or $T$.}
\renewcommand{\arraystretch}{2}
\tiny
\begin{tabular}{|c|c|c|c|c|c|c|c|}
\hline
\begin{tabular}[c]{@{}c@{}}\textbf{Dynamic} \\ \textbf{Graph Model}\end{tabular} &
\begin{tabular}[c]{@{}c@{}}\textbf{Initial} \\ \textbf{Config}\end{tabular} &
\textbf{$k$} &
\begin{tabular}[c]{@{}c@{}}\textbf{Necessary} \\ \textbf{Assumption}\end{tabular} &
\begin{tabular}[c]{@{}c@{}}\textbf{Sufficient} \\ \textbf{Assumption}\end{tabular} &
\multicolumn{3}{c|}{\textbf{Algorithm}} \\ \hline \hline

\begin{tabular}[c]{@{}c@{}}1-Interval \\ Connected\end{tabular} &
Scattered &
$n-1$ &
\begin{tabular}[c]{@{}c@{}}Global Comm, \\ 1-hop visibility\end{tabular} &
\begin{tabular}[c]{@{}c@{}}Global Comm, \\ 1-hop visibility\end{tabular} &
\multicolumn{3}{c|}{\begin{tabular}[c]{@{}c@{}}$\Theta(n)$ time, \\ $O(\log n)$ memory per agent\end{tabular}} \\ \hline

\begin{tabular}[c]{@{}c@{}}$T$-Path \\ Connected\end{tabular} &
Scattered &
$n-1$ &
\begin{tabular}[c]{@{}c@{}}Global Comm, \\ 1-hop visibility\end{tabular} &
\begin{tabular}[c]{@{}c@{}}Global Comm, \\ 1-hop visibility\end{tabular} &
\multicolumn{3}{c|}{\begin{tabular}[c]{@{}c@{}}$\Theta(n \cdot T)$ time, \\ $O(\log n)$ memory per agent\end{tabular}} \\ \hline

\begin{tabular}[c]{@{}c@{}}Connectivity \\ Time\end{tabular} &
Scattered &
$>n$ &
Open Question &
Open Question &
\multicolumn{3}{c|}{Open Question} \\ \hline \hline

\end{tabular}
\label{tb:algo-exp}
\end{table}

\subsection{Exploration Problem}
In this work, we present lower bounds on the number of agents required to solve the exploration problem across all three connectivity models. Additionally, we provide the necessary and sufficient conditions for solving the exploration problem in both 1-Interval Connected graphs and $T$-Path Connected graphs. The details of our key contributions are outlined as follows.

\begin{itemize}
    \item In Theorem \ref{thm:imp_Exp_1-Interval}, we have shown that a set of $k \leq n-2$ agents can't solve the exploration problem in the dynamic graphs, which hold the 1-Interval Connectivity property. This impossibility holds even if agents have infinite memory, full visibility, global communication, and know the parameters $k$, $n$.
    \item According to Theorem \ref{thm:exp-1-hop} and \ref{thm:exp-global}, solving the exploration problem with $n-1$ agents in 1-Interval Connected dynamic graphs requires 1-hop visibility and global communication unless they are in dispersed configuration. 
    \item Let $n-1$ agents be dispersed initially. It is impossible to solve the exploration of $n-1$ mobile agents on a 1-Interval Connected dynamic graph when they are equipped with global communication and unlimited memory but without 1-hop visibility (refer to Theorem \ref{thm:disp-Exp}).
    \item Any algorithm solving exploration problem in 1-Interval Connected Dynamic graph of $n$ nodes requires $\Omega(n)$ rounds (refer to Theorem \ref{thm:time_lower_exp}). And, any algorithm solving exploration problem in $T$-Path Connected graph of $n$ nodes requires $\Omega(n\cdot T)$ rounds (refer to Theorem \ref{thm:time_lower_exp1}). These results hold even if agents have infinite memory, full visibility, global communication, and know the parameters $k$, $n$, $T$.

    \item If the initial configuration contains at least two holes, then a group of \( k \leq n \) agents cannot solve the exploration problem in dynamic graphs that maintain the Connectivity Time property (refer to Theorem \ref{thm:connectivity_exp}). This impossibility holds even if agents have infinite memory, full visibility, global communication, and know the parameters $k$, $n$, $T$.

    \item We provide an algorithm which solves exploration with termination in 1-Interval connected dynamic graphs with $n-1$ agents in $\Theta(n)$ rounds using $O(\log n)$ bits of memory per agent in the synchronous setting with global communication and 1-hop visibility (refer to Section \ref{sec:EXP-1-Int}).

    \item  We provide an algorithm which solves perpetual exploration in $T$-Path Connected graphs with $n-1$ agents in $\Theta(n\cdot T)$ rounds using $O(\log n)$ bits of memory per agent in the synchronous setting with global communication and 1-hop visibility (refer to Section \ref{sec:EXP-T-Int}).
\end{itemize}

Refer to Table \ref{tb:algo-exp} for necessary and sufficient assumptions to solve the exploration problem in different dynamic graph models.

\section{Preliminaries}\label{sec:pre}
In this section, we recall the results from the existing literature which we use in our algorithm. In \cite{Ajay_dynamicdisp}, Kshemkalyani et al. study the dispersion problem on dynamic graphs with 1-Interval Connectivity, i.e., the adversary can add or delete edges arbitrarily maintaining $1$-Interval Connectivity. In their paper, the authors provide two impossibility results, one lower-bound along with an optimal algorithm. We denote their algorithm by $\mathcal{DISP}$. We use $\mathcal{DISP}$ as a sub-routine in our algorithm. Their results are as follows.

\begin{theorem}\label{thm:imp_1_nbhd}
     \cite{Ajay_dynamicdisp} It is impossible to solve the dispersion of $k \geq 5$ mobile agents on a dynamic graph deterministically with the agents having 1-hop visibility and unlimited memory but without global communication. 
\end{theorem}

\begin{theorem}\label{thm:imp_global}
    \cite{Ajay_dynamicdisp} It is impossible to solve the dispersion of $k \geq 3$ mobile agents on a dynamic graph deterministically with the agents having global communication and unlimited memory but without 1-hop visibility.
\end{theorem}


\begin{theorem}\label{thm:lower_bound}
   \cite{Ajay_dynamicdisp} Any algorithm solving the dispersion of $k \leq n$ agents on 1-Interval Connected dynamic graphs of $n$ nodes requires $\Omega(k)$ rounds. The lower bound holds even if the agents have unlimited memory.
\end{theorem}

\begin{theorem}\label{thm:main_result}
   \cite{Ajay_dynamicdisp} Given $k \leq n$ agents placed arbitrarily on the nodes of any $n$-node graph $\mathcal{G}_r$ that dynamically changes in every round $r \geq 0$ following the 1-Interval Connected dynamic graph model, $\mathcal{DISP}$ solves the dispersion in $\Theta(k)$ rounds with $\Theta(\log k)$ bits at each agent in the synchronous setting with global communication and 1-hop visibility. Also, all agents understand that dispersion has been achieved and terminate. 
\end{theorem}


\section{Impossibilities and Lower Bounds on Dispersion}\label{sec:existing}
In this section, we present lower bounds and impossibility results on $T$-Path Connected graphs and Connectivity Time graphs.

\begin{theorem}\label{thm:imp-connec}
     $(k\geq 3)$ A set of $k \leq n$ agents can't solve implicit dispersion in the dynamic graphs which hold the Connectivity Time property. This impossibility holds even if agents have infinite memory, full visibility, global communication, and know $k$, $n$, $T$. 
\end{theorem}

\begin{proof}
    Let $T\geq 2$, $n \geq 3$ and $k\leq n$. We show that there is no deterministic algorithm solving dispersion of $k \geq 3$ agents. Let algorithm $\mathcal{A}$ solve dispersion on dynamic graphs. Since agents are not dispersed, there are at least $n-k+1$ nodes which are holes initially. Define $p=n-k+1$. Let $v_1$, $v_2$, \ldots, $v_n$ be the nodes of $\mathcal{G}_0$. Consider the following configuration for each round $t\geq 0$.
    \begin{itemize}
        \item $\bm{t \in [0, \, T-2]}$: Without loss of generality, let $v_1$, $v_2$, \ldots, $v_p$ be holes initially. Nodes $v_1$, $v_2$, \ldots, $v_p$ form a star graph $S_1$, and the remaining nodes (i.e., $v_{p+1}$, $v_{p+2}$, \ldots, $v_n$) form a star graph $S_2$ (there may be holes in $S_2$). And, there is no edge between $S_1$ $\&$ $S_2$.
        \item $\bm{t \in [iT-1, \, (i+1) T-2]}$, \textbf{where} $\bm{i \geq 1\, \& \,i(}$\textbf{mod} $\bm{2) \neq 0}$: At the end of round $iT-2$, there are two star graphs $S_1$ $\&$ $S_2$. The star graph $S_1$ contains holes. Since there are $k$ agents in the star graph $S_2$ of size $k-1$ at the round $iT-2$, there is a multinode (say $v$) at the end of round $iT-2$ in $S_2$. In round $t$, the adversary makes node $v$ as an isolated node and forms a star of the remaining vertices (i.e. star graph $S$ of $\{v_1, v_2, \ldots, v_n\}-\{v\}$). 
        \item $\bm{t \in [i T-1, \, (i+1)T-2]}$, \textbf{where} $\bm{i \geq 1\, \& \,i(}$\textbf{mod} $\bm{2) = 0}$: At the end of round $iT-2$, there is one star graph $S$ and one isolated vertex $v$. Since at most $k-2$ agents are in $S$ at the end of round $iT-2$, there are at least $p$ vertices which are holes. Let $w_1$, $w_2$, \ldots, $w_p$ be holes in the star graph $S$ at the end of round $iT-2$. At the beginning of round $t$, the adversary forms a star graph $S_1$ of the holes $w_i$s, and forms the star graph $S_2$ of nodes $\{v_1, v_2, \ldots,v_n\}\backslash\{w_1, w_2, \ldots, w_p\}$. And, there is no edge between $S_1$ $\&$ $S_2$.
    \end{itemize}

    It is not hard to observe that for each round $t$, the above construction holds the Connectivity Time property. In the construction described above, solving the dispersion problem is impossible for the following reason: If \( t \in [0, T-2] \), the dispersion problem cannot be solved because \( k \) agents are arranged in a star graph \( S_2 \) with \( k-1 \) nodes. Consequently, there exists a multinode in the graph \( \mathcal{G}_t \) for all \( t \in [0, T-2] \). 

Furthermore, for any integer \( i \geq 2 \) where \( i \mod 2 \neq 0 \), if \( t \in [iT-1, (i+1)T-2] \), solving the dispersion is also impossible since a multinode becomes an isolated node during round \( t \). For any integer \(i \geq 2\) that is even (i.e., \(i \mod 2 = 0\)), and for any time \(t\) within the interval \([i T - 1, \; (i + 1) T - 2]\), it becomes impossible to solve the dispersion problem at the end of round \(t\) for the following reason. At the end of round $iT-2$, at most \(k - 2\) agents are present in the star graph \(S\) of size \(n - 1\). Therefore, at least \(p+1\) holes within the star graph \(S\) at the end of round $iT-2$. At the start of round \(iT-1\), the adversary creates a star graph \(S_1\) with \(p\) holes, while the remaining nodes form another star graph \(S_2\). In star graph \(S_2\), there are \(k - 1\) nodes and \(k\) agents, which makes dispersion impossible to solve at the end round $iT-1$. Since, the adversary maintains this configuration for every round \(t\in[i T - 2, \; (i + 1) T - 2]\), where $i$ is even, therefore due to the same reason, dispersion is impossible to solve at round $t$. No additional advantages, such as infinite memory, full visibility, global communication, and knowledge of the parameters $k$ and $n$, help the agents to solve dispersion. This completes the proof. 
\end{proof}


It is impossible to solve the implicit dispersion in 1-Interval Connected graphs using either only 1-hop visibility or only global communication due to Theorem \ref{thm:imp_1_nbhd}, \ref{thm:imp_global}. Also, $1$-Interval connected graphs are $T$-Path Connected graphs (refer Observation \ref{obs:connectivity}). Therefore, we have the following observation.

\begin{observation}
   Without 1-hop visibility and global communication, it is impossible to solve the dispersion problem (Implicit/Explicit) in $T$-Path Connected graphs.
\end{observation}

As discussed in Section \ref{sec:related}, to solve explicit dispersion in $1$-Interval Connected graphs, global communication is enough for agents to terminate once dispersion is achieved. In the $T$-Path Connectivity model, we ask the following question.

\begin{question}\label{ques:necessary}
   Is it possible to solve explicit dispersion\footnote{We discuss implicit dispersion in Section \ref{sec:algorithm} as Remark 2} when agents have the same capability as in \cite{Ajay_dynamicdisp}, i.e., agents are equipped with global communication, 1-hop visibility but do not know $n$, $k$? 
\end{question}

 The answer is no due to the following impossibility result.



\begin{figure}
    \centering
    \includegraphics[width=0.75\linewidth]{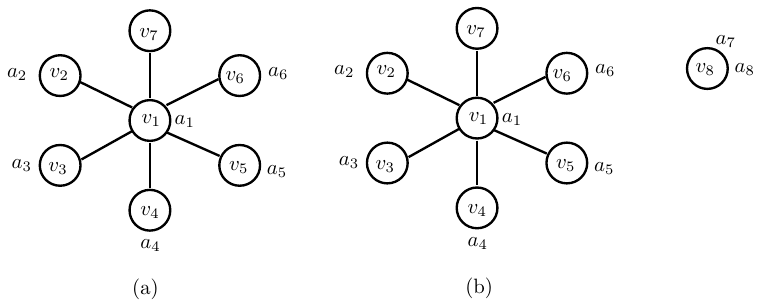}
     \caption{(a) $\mathcal{G}$ for $n=7$ with agents' position, (b) $\mathcal{G'}$ for $n=7$ with agents' position}
    \label{fig:Image_Imp1}
\end{figure}

\begin{figure}
    \centering
    \includegraphics[width=0.75\linewidth]{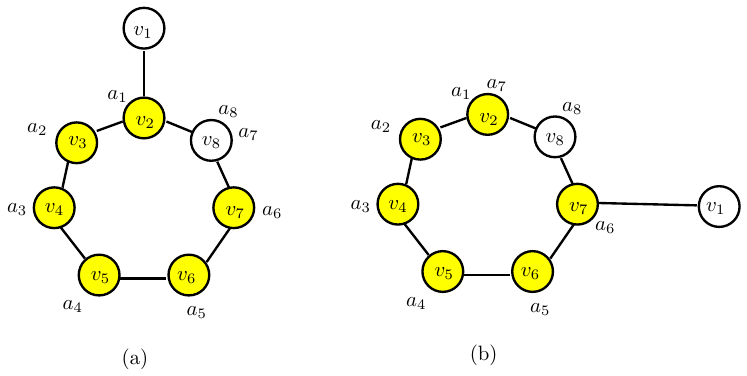}
      \caption{(a) $\mathcal{G}_{t+1}$ for $n=7$ with agents' position, (b) $\mathcal{G}_r$, $r>t+1$, for $n=7$ with agents' position}
    \label{fig:Image_Imp2}
\end{figure}

\begin{theorem}\label{thm:imp:knowledge_T}
(For $n\geq 3$)  If agents are equipped with global communication and full visibility and do not know $n$, $k$, $T$, then it is impossible to solve explicit dispersion in $T$-Path Connected graphs.
\end{theorem}

\begin{proof}
    Let $\mathcal{A}$ be an algorithm which solves explicit dispersion in $T$-Path Connected graph $G$. Let $v_1$, $v_2$, $v_3$, \ldots, $v_{n-1}$, $v_{n}$ be nodes of $G$. Initially, agents $a_1$, $a_2$, \ldots, $a_{n-1}$ are at nodes of $G$ in the following manner. Agent $a_i$ is at node $v_i$ for all $i \in [1, \,n-1]$. The initial configuration of $G$ (i.e., $\mathcal{G}_0$) is a star graph of these nodes. For each round $r$, $\mathcal{G}_r$ is $\mathcal{G}_0$ (refer Fig. \ref{fig:Image_Imp1}(a)). The graph $G$ maintains the $T$-Path Connectivity. The agent $a_i$ has no information of parameter $n$, $k$ and $T$. If algorithm $\mathcal{A}$ solves the dispersion in $G$, then agents terminate in some round $t$ after achieving dispersion. Since $n-1$ agents are in $n$ node dynamic graph $G$, there is a hole at the end of round $t$. Without loss of generality, node $v_1$ is a hole in $G$ at the end of round $t$.

    We construct a new $T'$-Path Connected graph $G'$ of $n+1$ nodes. Let $v_1$, $v_2$, $v_3$, \ldots, $v_{n-1}$, $v_{n}$, $v_{n+1}$ be nodes of $G'$. Initially, agents $a_1$, $a_2$, \ldots, $a_{n-1}$, $a_n$, $a_{n+1}$ are at nodes of $G'$ in the following manner. Agent $a_i$ is at node $v_i$ for all $i \in [1, \,n-1]$, and agent $a_{n}$, $a_{n+1}$ are at node $v_{n+1}$. Let $T'$ be $t+2$. The construction of $\mathcal{G}'$ of $G'$ in the first $T'-2$ rounds is as follows. For each $r \in [1, \,t]$, nodes $v_1$, $v_2$, \ldots, $v_{n-1}$ forms $\mathcal{G}_0$, and node $v_{n+1}$ is not connected with $v_i$, $1\leq i \leq n-1$ (refer Fig. \ref{fig:Image_Imp1}(b)). In the beginning of round $T'-1$ (i.e., round $t+1$), the adversary forms a cycle of nodes $v_2$, \ldots, $v_{n}$, $v_{n+1}$, and connects $v_2$ with node $v_1$ (refer Fig. \ref{fig:Image_Imp2}(a)). It is trivial to observe that $\mathcal{G}'_r$ maintains $T'$-Path Connectivity in round $r \in [1, \, T']$ in $G'$. In the first \( T' - 2 \) rounds, agent \( a_i \) (where \( i \in [1, n-1] \)) is unable to understand the existence of agents \( a_n \) and \( a_{n+1} \), despite having full visibility and global communication. This is because agents \( a_n \) and \( a_{n+1} \) are located in a different connected component. Also, agents have no information of $n$ and $k$. Therefore agent $a_i$, $i \in [1,\,n-1]$ is not able to differentiate between $G$ and $G'$ for the first $T'-2$ rounds and terminate in $G'$ as $T'-2=t$.  
    
    At the round $t+1$, node $v_1$ is a hole. At the beginning of round $t+1$, node $v_1$ is at least two hops away from node $v_{n+1}$ (refer Fig. \ref{fig:Image_Imp2}(a)). In the round $t+1$, agents $a_{n-1}$, $a_n$ do not get any help from the other agents as the other agents are already terminated. What agents $a_{n-1}$, $a_n$ can do is to move at most 1-hop at any direction at the end of round $t+1$. Therefore, by the start of round $t+2$, the node $v_1$ remains as a hole. Further, we need to show that $v_1$ remains a hole for each round $r$, $r>t+2$. In every round $r$, the adversary maintains $\mathcal{G}'_r$ is connected. Therefore, $G'$ maintains $T'$-Path connectivity for every round $r$. At the beginning of round $r$ ($r>t$), the adversary forms $\mathcal{G}'_r$ in the following way. Let agent $a_{n}$ and $a_{n+1}$ be at node $v_i$ and $v_j$, respectively, where $i, j \in [2, \, n+1]$. The adversary forms a cycle of nodes $v_2$, \ldots, $v_{n-1}$, $v_{n}$, and add an edge between node $v_1$ and $v_k$, where $k \neq i,j\, \& \,k \in [2, \, n+1]$ (refer Fig. \ref{fig:Image_Imp2}(b)). It is not difficult to observe that $v_1$ is at least two hops away from node $v_i$ and $v_j$. Since agent $a_n$, $a_{n+1}$ can move at most 1-hop, therefore at the end of round $r$, the node $v_1$ remains a hole. Note that global communication and full visibility can not help agents to achieve dispersion. Therefore, for every round $r$ ($r>t$), the node $v_1$ remains the hole, and the dispersion is not achieved in $G'$. This completes the proof.
\end{proof}

Due to Theorem \ref{thm:imp:knowledge_T}, agents need the knowledge of $T$ if agents have the same capability as in \cite{Ajay_dynamicdisp}. Now, we show a time lower bound of $\Omega(k\cdot T)$ rounds to solve implicit dispersion in the $T$-Path Connected graphs $G$ \footnote{A time lower bound for implicit dispersion also applies to explicit dispersion as without solving implicit dispersion, one cannot solve explicit dispersion.} considering agents know all the parameters. The result is as follows. 

\begin{figure}[ht]
    \centering
    \includegraphics[width=0.5\linewidth]{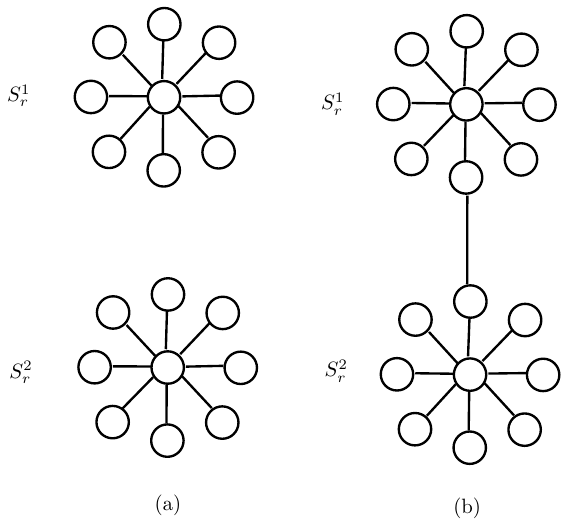}
 \caption{(a) $\mathcal{G}_r$ at round $r$, where $p(T-1)+1 \leq r < (p+1)(T-1)$, (b) $\mathcal{G}_r$ at round $r$, where $r=(p+1)(T-1)$}
     \label{fig:lower}
\end{figure}

\begin{theorem}\label{thm:lower_bound_T}
Any algorithm solving implicit dispersion on any $T$-Path Connected graph of $n$ nodes requires $\Omega(k \cdot T)$ rounds. The lower bound holds even if the agents have infinite memory, are equipped with global communication, have full visibility and know all of $k$, $n$, $T$. This proof is valid when the dynamic diameter of the tree is $\hat{D} = O(1)$.
\end{theorem}


\begin{proof} 
    Let $G$ be $T$-Path Connected graph. We show that the adversary can construct $\mathcal{G}_r$, $r>0$, i.e., a sequence of dynamic graphs such that after every $T$ round at most one new node can be visited by the agents. We construct a dynamic forest on which the agent requires at least $k\cdot T$ rounds to visit $k$ (new) nodes by the agents. Let, at round $r$, set $V^1_r$ consists of some nodes of $G$ which contain at least one agent, and $V^2_r=V\backslash V^1_r$. Let $S^1_r$ be a star graph formed by the nodes of $V^1_r$, and $S^2_r$ be a star graph formed by the nodes of $V^2_r$. Initially, all agents are co-located at node $v$ and we consider $V^1_1=\{v\}$. If $r<T-1$, then there is no edge between nodes of $S^1_r$ and $S^2_r$. If $r=T-1$, then the adversary connects $S^1_r$ and $S^2_r$ by some edge $e$. In round $r=T-1$, agents may access a hole. Therefore, in round $T-1$, all the agents are in at most two nodes. In round $T \leq r < 2(T-1)$, the adversary keeps all the nodes with agents in $V^1_r$, i.e., the size of $V^1_r$ is at most two. In the same manner, the adversary connects $S^1_r$ and $S^2_r$ via some edge in round $2(T-1)$. Therefore, agents may access another hole in round $r=2(T-1)$, and agents are at most three nodes in round $r=2(T-1)$. Again in round $2T \leq r < 3(T-1)$, the adversary keeps all the nodes with agents in $V^1_r$, i.e., the size of $V^1_r$ is at most three. Continuing in the same way, in round $r$ where $p(T-1)+1 \leq r < (p+1)(T-1)$ for any $p\in [1,k-1]$, the adversary keeps all the nodes with agents in set $V^1_r$ (i.e., the size of $V^1_r$ is at most $k-1$) and maintains no edge between nodes of $S^1_r$ and $S^2_r$ (refer Fig. \ref{fig:lower}(b)). Hence till round $k(T-1)-1$, the agents can not reach a dispersed configuration. For all rounds $r\geq k(T-1)$, $S^1_r$ and $S^2_r$ is connected via an edge (refer Fig. \ref{fig:lower}(b)). In this way, the adversary maintains $T$-Path Connectivity in $G$. This completes the proof. 
\end{proof}

\section{Dispersion in $T$-Path Connected Graphs}
In this section, we present an algorithm that takes $O(k \cdot T)$ time to solve the explicit dispersion problem for $k \leq n$ in an arbitrary $n$ node $T$-Path Connected anonymous graph when agents are equipped with global communication, 1-hop visibility and knowledge of $T$. For the sake of completeness, before we give an overview of our algorithm, we give a high-level idea of the algorithm $\mathcal{DISP}$ of \cite{Ajay_dynamicdisp}.\\

\noindent \textbf{High-level idea of $\mathcal{DISP}$ \cite{Ajay_dynamicdisp}:} The idea is to slide the agents from multinode(s) to hole(s), in every round $r$ until all agents are dispersed. Consider a path $P$ = $v_1$, $v_2$,\ldots,$v_{l-1}$, $v_l$ such that $v_1$ is a multinode, $v_2$,...,$v_{l-1}$ has an agent each, and $v_l$ is an empty node. Given path $P$, sliding means that an agent from node $v_j$ moves to node $v_{j+1}$, for $1 \leq j < l$. By virtue of global communication, an agent $a_i$ finds the information of such a hole from the agent at node $v_{l-1}$ as the node $v_l$ is in the 1-hop neighbourhood of $v_{l-1}$ (refer Fig. \ref{fig:case-1}). After this sliding, node $v_l$ (in path $P$), which previously was empty, now has an agent(s) positioned on it. It may be possible there are multiple paths leading to the same hole (refer Fig. \ref{fig:case-2}). In this case, even if multiple agents reach the same hole at some round $r$ from different paths, there is one less hole from the previous round. One more case is possible, which is as follows. Consider a path $P_1=u_1$, $u_2$,\ldots,$u_{l}$, and each node $u_i$ for $2\leq i \leq l-1$ contains one agent. Nodes $u_1$ and $u_l$ are multinode. For some $2\leq j \leq l-1$, $u_j$ is connected to a hole $v$. Therefore, there are two paths that lead to hole $v$ (refer Fig. \ref{fig:case-3}). Their algorithm takes care of this case by making the agents slide through only one such path, i.e., either through the path $u_1$, $u_2$,\ldots,$u_{j}$, $v$ or through $u_l$, $u_{l-1}$,\ldots,$u_{j}$, $v$. All this can be done in one single round.\\     
\\

\begin{figure}[htbp]
    \centering
    \begin{minipage}[b]{1\textwidth}
        \centering
        \includegraphics[width=0.75\linewidth]{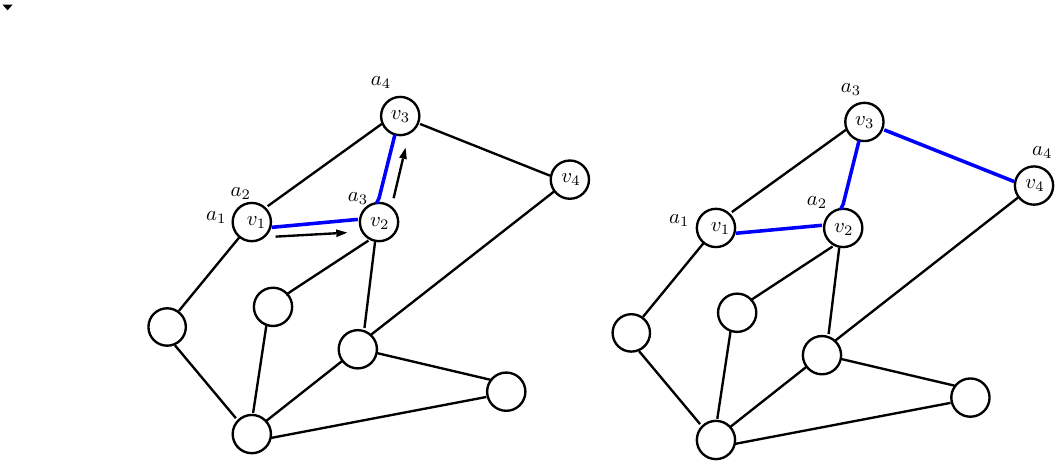}
        \caption{$v_1$, $v_2$, $v_3$, $v_4$ form a path where $v_1$ is a multinode and $v_4$ is a hole. The path through which agents are sliding is shown using $\rightarrow$.}
        \label{fig:case-1}
    \end{minipage}
    \hfill
    \begin{minipage}[b]{1\textwidth}
        \centering
        \includegraphics[width=0.75\linewidth]{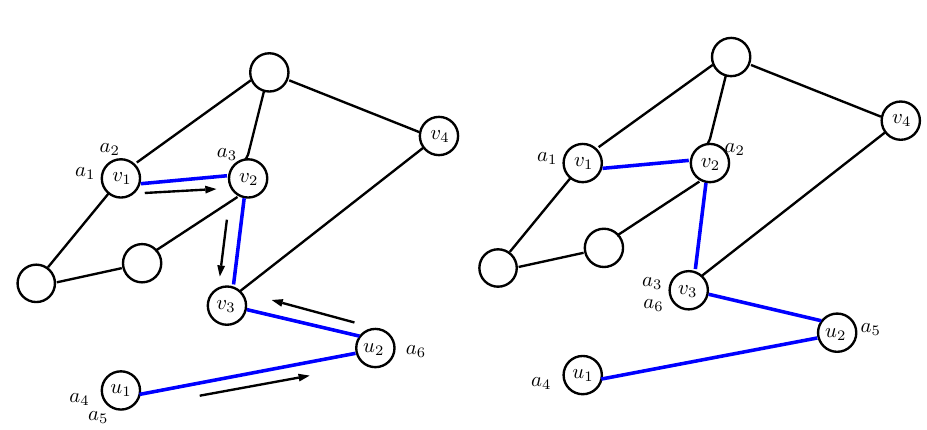}
        \caption{$v_1$, $v_2$, $v_3$ is a path where $v_1$ is a multinode and $v_3$ is a hole, and $u_1$, $u_2$, $v_3$ is a path where $u_1$ is a multinode. The path through which agents are sliding is shown using $\rightarrow$.}
        \label{fig:case-2}
    \end{minipage}
\end{figure}
\begin{figure}
    \begin{minipage}[b]{1\textwidth}
        \centering
        \includegraphics[width=0.75\linewidth]{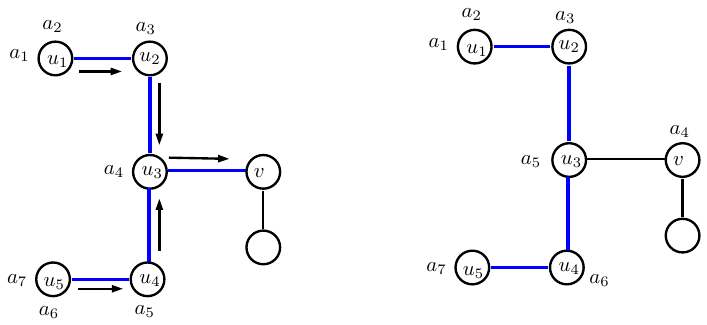}
        \caption{$u_1$, $u_2$, $u_3$, $u_4$, $u_5$ is a path where $u_1$ and $u_5$ are multinodes. Node $v$ is a hole connected to $u_3$, which is one hop away from the hole $v$. The path through which agents are sliding is shown using $\rightarrow$.}
        \label{fig:case-3}
    \end{minipage}
\end{figure}

\noindent \textbf{Overview of our algorithm:} Our idea to solve explicit dispersion in $T$-Path Connected graphs is similar to \cite{Ajay_dynamicdisp}. We use $\mathcal{DISP}$ as a subroutine in our algorithm. Recall that, in $T$-Path Connectivity, between any two nodes, there exists a path in at least one round within any consecutive $T$ rounds. If there is a hole and a multinode in $G$ at round $r$, then at some round $t \in [r, \,r+T-1]$, there must exist a path between the hole and the multinode. In this case, agents can fill the hole with the help of $\mathcal{DISP}$ in that round. In \cite{Ajay_dynamicdisp}, the agents can understand the termination with the help of global communication as the $\mathcal{G}_t$ is connected in every round. In $T$-Path Connected graphs, we can not rely on the global communication as $\mathcal{G}_t$ might be disconnected in every round $t$. Therefore, an agent needs to understand that the dispersion is achieved without the knowledge of $k$, $n$. Here, we use the fact that $T$ is known to the agents.

\begin{algorithm}[H]
    \caption{Explicit Dispersion}
    \label{algorithm:T-Path-disp}
    $a_i.t=0$\\
    \While{True}
    {
        agent $a_i$ broadcasts 1-hop neighbours information, $a_i.ID$ and $a_i.count$ \\
        \If{$a_i.count>1$}
        {
            $a_i.t=0$\\
            \If{$a_i$ receives $a_j.count>1$ from some agent $a_j$}
            {
                \If{$a_i$ has a hole in 1-hop or it receives a hole from some agent $a_k$}
                {
                    \If{$a_i$ is minimum ID agent among $a_j$}
                    {
                        it computes sliding path as per $\mathcal{DISP}$
                    }
                    \ElseIf{$a_i$ is not minimum ID agent among $a_j$}
                    {
                        $a_i$ computes a sliding path based on minimum ID agent among $a_j$. It moves if $a_i$ is on the sliding path as per $DISP$ of $a_j$. Otherwise, it waits at its position.
                    }
                }
                \ElseIf{$a_i$ has no hole in 1-hop \& it does not receive a hole from some agent $a_k$}
                {
                    it stays at its position
                }
            }
            \If{$a_i$ does not receive $a_j.count>1$ from some agent $a_j$}
            {
                \If{$a_i$ has a hole in 1-hop or it receives a hole from some agent $a_k$}
                {
                    it computes sliding path as per $\mathcal{DISP}$
                }
                \ElseIf{$a_i$ has no hole in 1-hop $\&$ it does not receive a hole from some agent $a_k$}
                {
                    it stays at its position
                }
            }
        }
        \ElseIf{$a_i.count=1$}
        {
            \If{it receives $a_j.count>1$ from some agent $a_j$}
            {
                $a_i.t=0$
                \If{$a_i$ has a hole in 1-hop or it receives a hole from some agent $a_k$}
                {
                    $a_i$ computes a sliding path based on minimum ID agent among $a_j$. It moves if $a_i$ is on the sliding path as per $DISP$ of $a_j$. Otherwise, it waits at its position.
                }
                \ElseIf{$a_i$ has no hole in 1-hop \& it does not receive a hole from some agent $a_k$}
                {
                    it stays at its position
                }
            }
            \ElseIf{it does not receives $a_j.count>1$ from some agent $a_j$}
            {
                 $a_i.t=a_i.t+1$\\
                \If{$a_i.t=T$}
                {
                    $a_i$ terminates
                }
                \Else
                {
                    stay at the current node
                }
            }    
        }
    }
\end{algorithm}

\subsection{The Algorithm}\label{sec:algorithm}
Agents maintain the following parameters. 
\begin{itemize}
    \item $a_i.ID$: It represents the ID of agent $a_i$.
    \item $a_i.count:$  It represents the number of agents present at the node. If $a_i.count=1$, then the current node where $a_i$ resides is not a multinode; else, if $a_i.count>1$, then the current node is a multinode. 
    \item $a_i.t$: Agent $a_i$ uses this to count the rounds.
\end{itemize}

A detailed description of the algorithm is as follows. At round $r$, $\mathcal{G}_r$ may be disconnected. In each round $r$, if agent $a_i$ is the minimum ID at the current node, then it broadcasts the information of its 1-hop along with $a_i.count$ and $a_i.ID$. Based on this information, agent $a_i$ does the following.
    \begin{enumerate}
        \item If agent $a_i.count>1$ and it receives $a_j.count>1$ from at least one other agent $a_j$, it updates $a_i.t=0$ and works as mentioned below. 
        \begin{itemize}
            \item If $a_i$ receives a hole information, then it performs the following steps.
            \begin{itemize}
                \item If $a_i$ is minimum ID among $a_j$s, then it runs $\mathcal{DISP}$.
                \item If $a_i$ is not minimum ID among $a_j$s, then it moves as follows. Let $a_j$ be the minimum ID agent. If $a_i$ is on the sliding path as per $\mathcal{DISP}$ of $a_j$, then it moves. Otherwise, it stays at its position.
            \end{itemize}
            \item If $a_i$ does not receive a hole information, then it stays at its position.
        \end{itemize}
        \item If agent $a_i.count>1$ and it does not receive $a_j.count>1$ from some agent $a_j$, it makes the following decision. It updates $a_i.t=0$. 
        \begin{itemize}
            \item If $a_i$ receives a hole information, then it runs $\mathcal{DISP}$.
            \item If $a_i$ does not receive a hole information, then it stays at its position.
        \end{itemize}
        \item If agent $a_i.count=1$ and it receives $a_j.count>1$ from some agent $a_j$, then it makes the following decision. It updates $a_i.t=0$.
        \begin{itemize}
            \item Let $a_k$ be minimum ID agent among $a_j$s. If $a_i$ receives a hole information or there is a hole in 1-hop of $a_i$, and $a_i$ is on the sliding path based on $\mathcal{DISP}$ of $a_j$, then it moves on the sliding path.
            \item If $a_i$ does not receive the hole information and there is no hole in its neighbours, it stays at its position.
        \end{itemize}

        \item If agent $a_i.count=1$ and it does not receive $a_j.count>1$ from some agent $a_j$, then it updates $a_i.t=a_i.t+1$, and compare the value of $a_i.t$ and $T$. If $a_i.t=T$, then it terminates. Otherwise, it stays at its current position.  
    \end{enumerate}

Refer to Algorithm \ref{algorithm:T-Path-disp} for the pseudocode.
\subsection{Correctness and Analysis of the Algorithm}\label{sec:analysis-T-Path}
At round $r(\geq 0)$, there can be more than one multinode in graph $\mathcal{G}_r$. Let $G_i$ be a connected component of $\mathcal{G}_r$. Consider $G_i$ contains at least one multinode and at least one hole. Let $u_1$, $u_2$, \ldots, $u_l$ be multinodes in $G_i$, and $a_i$ be the minimum ID agent at node $u_i$. Without loss of generality, let $a_1$ be the minimum ID agent among all $a_i$s. Algorithm \ref{algorithm:T-Path-disp} does the following step: using $\mathcal{DISP}$, $a_1$ fix a sliding path $P=w_1(=u_1)\sim w_2 \sim \ldots \sim w_p$ in $G_i$ such that $w_j$, $j\in [2,p-1]$ contains at least one agent, and node $w_p$ is a hole. Let $b_j$ be the minimum ID agent at node $w_j$, $j\in[2, p-1]$. In round $r$, $a_1$ moves from node $w_1$ to $w_2$, and agent $b_j$ moves from node $w_j$ to node $w_{j+1}$, where $j\in[1,p-1]$. In this way, by sliding, a hole is filled. As a consequence, the number of agents in $v_1$ decreases by 1, the hole gets one agent, and the number of agents in each $w_j$, $j\in[2, p-1]$ remain the same.  This leads to the following observation. 

\begin{observation}\label{obs:multinodes}
    Let $l$ be the number of multinode at round $r$. The number of multinode at round $r'>r$ is less than or equal to $l$.
\end{observation}

\begin{lemma}\label{lm:termination_cond}
    No agent terminates before the dispersion is achieved. 
\end{lemma}
\begin{proof}
    Let \( a_1 \) be the agent at node \( v \), and let it terminate at round \( r \), meaning \( a_i.t = T \) at round \( r \). If the dispersion is not achieved by the end of round \( r \), then there is at least one multinode present at that time. Let \( w \) represent a multinode in \( \mathcal{G}_r \) at round \( r \). According to Observation \ref{obs:multinodes}, the number of multinodes does not increase; therefore, node \( w \) is also a multinode in \( \mathcal{G}_{r'} \) for every \( r' \) within the interval \([0, r]\). Since \( a_1.t = T \) at round \( r \), nodes \( v \) and \( w \) cannot be in the same component of \( \mathcal{G}_{r''} \), where \( r'' \) falls within the range \([r-T+1, r]\). This scenario is not possible, as it contradicts the definition of \( T \)-Path connectivity. This completes the proof.
\end{proof}

\begin{lemma}\label{lm:correctness}
    Algorithm \ref{algorithm:T-Path-disp} solves the dispersion problem.
\end{lemma}
\begin{proof}
    At some round $r$, if there is at least one multinode in $\mathcal{G}_{r}$, then there is at least one hole in $\mathcal{G}_r$. If at least one multinode and at least one hole are in the same connected component (say $G_1$) of $\mathcal{G}_r$, then the agents in $G_1$ can understand the hole and multinode with the help of global communication and 1-hop visibility and thus the hole gets filled. Otherwise, if at least one multinode and at least one hole are not in the same connected component of $\mathcal{G}_r$, then due to $T$-Path Connectivity, there exists $j$ between rounds $r$ $\&$ $r+T-1$,
     such that there is one connected component of $\mathcal{G}_j$ which contains at least one hole and at least one multinode. Due to Lemma \ref{lm:termination_cond}, all agents are active at round $j$. Therefore, if the number of holes is not reduced by 1 till round $j-1$, then the number of holes is reduced by one in round $j$ by virtue of the $\mathcal{DISP}$ algorithm. Therefore, if $k \leq n$ agents are positioned on $p$ nodes of $\mathcal{G}_{r}$ at the beginning of round $r$, then by the end of round $r+T-1$, the agents are positioned on at least $p + 1$ nodes of $\mathcal{G}_{r+T-1}$. Since no agent terminates due to Lemma \ref{lm:termination_cond}, it continues till there is no multinode in $G$. Therefore, all agents achieve the dispersion.
\end{proof}

\begin{lemma}\label{lm:termination}
Agents terminate successfully in Algorithm \ref{algorithm:T-Path-disp}.
\end{lemma}
\begin{proof}
 In Algorithm \ref{algorithm:T-Path-disp}, whenever an agent finds no information of multinode, it increases the value of $a_i.t$, and whenever it finds information of multinode, it resets $a_i.t$ to 0. Due to Lemma \ref{lm:termination_cond}, no agent terminates before the dispersion is achieved, and Lemma \ref{lm:correctness} shows that the dispersion is achieved. Therefore, whenever the dispersion is achieved, then within the next $T$ rounds, the value of $a_i.t$ becomes $T$ for every agent $a_i$, and all the agents terminates.
\end{proof}

\begin{lemma}\label{lm:time}
    The run time of Algorithm \ref{algorithm:T-Path-disp} is $O(k \cdot T)$. 
\end{lemma}
\begin{proof}
    Due to the proof of Lemma \ref{lm:correctness}, if there is a multinode $v_i$ and there are holes in $\mathcal{G}_{r}$, then within $r+T$ rounds, the number of holes is decreased by at least 1. Since $k$ agents are present in $G$ initially, and in every $T$ round, at least one hole gets filled, agents may need $k\cdot T$ rounds to achieve dispersion. After this, agents may need further $T$ rounds for termination due to the proof of Lemma \ref{lm:termination}. Hence, the run time of Algorithm \ref{algorithm:T-Path-disp} is $O(k \cdot T)$.   
\end{proof}

\begin{lemma}\label{lm:memory}
    Agents require $O(\log \max(T, \, k))$ memory\footnote{Memory requirement can be improved when the algorithm does not require $T$ and work with either $k$ or $n$.} to run Algorithm \ref{algorithm:T-Path-disp}.
\end{lemma}
\begin{proof}
    Due to Theorem \ref{thm:main_result}, $\Theta(\log k)$ memory per agent is required to solve dispersion in 1-Interval Connected graphs. Since 1-Interval Connectivity is stronger than $T$-Path connectivity, therefore $\Omega(\log k)$ memory per agent is a must to solve dispersion in $T$-Path Connected graphs. The computation within a round happens in temporary memory. In our algorithm, agents do not remember 1-hop information and the information of global communication, and $\mathcal{DISP}$ requires $O(\log k)$ memory per agent. In our algorithm, an agent remembers count $a_i.t$ that can go up to $T$. Therefore, an agent needs $O(\log T)$ memory to store $a_i.t$. Hence, agents require $O(\log \max(k, \,T))$ memory to run Algorithm \ref{algorithm:T-Path-disp}. 
\end{proof}
Using Theorem \ref{thm:lower_bound_T} and Lemma \ref{lm:correctness}, \ref{lm:termination}, \ref{lm:time}, \ref{lm:memory}, we have the following theorem. 
\begin{theorem}
    Our algorithm solves explicit dispersion for $k \leq n$ agents in $\Theta(k\cdot T)$ rounds using $O(\log \max(k, \, T))$ bits of memory per agent in the synchronous setting with global communication, 1-hop visibility and knowledge of $T$.
\end{theorem}
Based on Algorithm \ref{algorithm:T-Path-disp}, we have two remarks as follows.
\begin{remark}
     To solve implicit dispersion, agents do not need the information $T$. We can modify Algorithm \ref{algorithm:T-Path-disp}. In Algorithm \ref{algorithm:T-Path-disp}, agents use parameter $T$ to achieve termination. To solve implicit dispersion, agents do not need parameter $T$. In other terms, whenever agents get information of hole and multinode in the same round, agents try to fill that hole using algorithm $\mathcal{DISP}$.
\end{remark}

\begin{remark}
    In our model, the ID range of agents is between $[1,k]$. If the ID range is extended to $[1,n^c]$, where $c$ is a constant, to remember its ID, the agent requires $\Omega(\log n)$ memory. In each round of Algorithm \ref{algorithm:T-Path-disp}, each agent $a_i$ remembers its ID and $a_i.t$. Therefore, Algorithm \ref{algorithm:T-Path-disp} solves explicit dispersion for $k \leq n$ agents in $\Theta(k\cdot T)$ rounds using $O(\log \max(n, \, T))$ bits of memory per agent in the synchronous setting with global communication, 1-hop visibility and knowledge of $T$.
\end{remark}
 \section{Impossibility and Lower Bounds on Exploration}
In this section, we study the exploration problem of all three connectivity models. In other words, we try to understand how many agents are necessary to solve exploration in all three models, and we also study the time required by the agents to solve the exploration problem in 1-Interval Connected graphs and $T$-Path Connected graphs.

\begin{figure}[h!]
    \centering
    \includegraphics[width=0.5\linewidth]{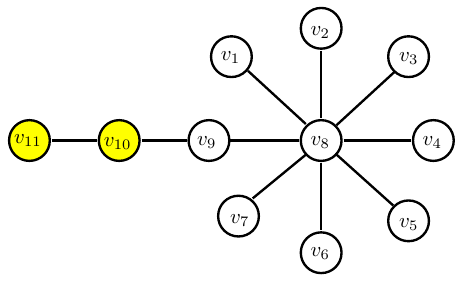}
      \caption{This figure is for $n=11$, and yellow shaded nodes are holes.}
    \label{fig:EXP_IMP}
\end{figure}



        
        
        
    

\begin{theorem}\label{thm:imp_Exp_1-Interval}
     A set of $k \leq n-2$ agents can't solve the exploration problem in the dynamic graphs, which hold the 1-Interval Connectivity property. This impossibility holds even if agents have infinite memory, full visibility, global communication, and know the parameters $k$, $n$.
\end{theorem}
\begin{proof}
    Let $V=\{v_1$, $v_2$, \ldots, $v_n\}$ be the set of nodes in $G$. Since $k \leq n-2$, at least two nodes are holes in the initial configuration. Without loss of generality, let $v_{n-1}$, $v_{n}$ be hole in $G$. We give a strategy for the adversary so that no agent can visit node $v_n$ in any round $r\geq0$. 

    At the beginning of the round $r=0$, the adversary forms a star graph of nodes $v_1$, \ldots, $v_{n-2}$, and forms a path of length 1 with $v_{n-1}$ and $v_n$. The adversary connects the node $v_{n-1}$ and $v_{n-2}$ by an edge. For $n=11$, we can see the configuration in Fig. \ref{fig:EXP_IMP}. In round $r=0$, no matter how agents move, the agents can not visit node $v_n$ due to the fact that node $v_n$ is at least 2-hop away from agents. Therefore, at the end of round $r=0$, node $v_n$ remains unvisited. At the beginning of round $r\geq1$, since at most $n-2$ agents are present, node $v_n$ and some node $v\in \{v_1,\, v_2,\, \ldots,\, v_{n-1}\}$ are holes in $G$. In this case, the adversary forms a star graph $\mathcal{S}_r$ of nodes $V-\{v,\,v_n\}$ and forms a 1-length path of nodes $v$ and $v_n$. At the beginning of round $r$, the adversary adds one edge between any node of $S_r$ and node $v$. In round $r\geq 1$, no matter how agents move, the agents can not visit node $v_n$ at the end of round $r$ due to the fact that node $v_n$ is at least 2-hop away from agents. Since, in every round $r\geq 0$, the graph is always connected, it satisfies the 1-Interval Connectivity property. Therefore, a set of $k \leq n-2$ agents can't solve the exploration problem in the dynamic graphs, which hold the 1-Interval Connectivity property. No additional advantages, such as infinite memory, full visibility, global communication, and knowledge of the parameters $k$ and $n$, help the agents to solve exploration. This completes the proof. 
\end{proof}
Based on Theorem \ref{thm:imp_Exp_1-Interval}, we have the following observation and remark.
\begin{observation}
    To solve the exploration problem in 1-Interval Connected graphs, a team of at least n-1 agents is necessary.
\end{observation}
\begin{remark}\label{rk:T-Path Connected}
    Since $1$-Interval Connectivity is stronger than $T$-Path Connectivity and Connectivity Time model, a set of $k \leq n-2$ agents can't solve the exploration problem in the dynamic graphs, which hold the $T$-Path Connectivity property and Connectivity Time Property. This impossibility holds even if agents have infinite memory, full visibility, global communication, and know the parameters $k$, $n$.
\end{remark}

Now, we present two theorems which show that agents need 1-hop visibility and global communication to solve the exploration with $n-1$ agents.

\begin{theorem}\label{thm:exp-1-hop}
    (For $n\geq 7$) It is impossible to solve the exploration with $n-1$ mobile agents on a 1-Interval Connected dynamic graph when the agents have 1-hop visibility and unlimited memory but without global communication, unless they start in a dispersed configuration.
\end{theorem}

\begin{figure}
    \centering
    \includegraphics[width=1\linewidth]{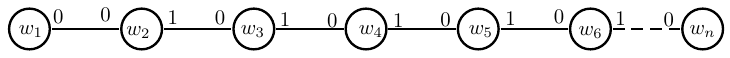}
    \caption{$\mathcal{P}_r$, the adversary supposed to give $\mathcal{P}_r$ at round $r$.}
    \label{fig:exp-1-hop}
\end{figure}

\begin{figure}
    \centering
    \includegraphics[width=1\linewidth]{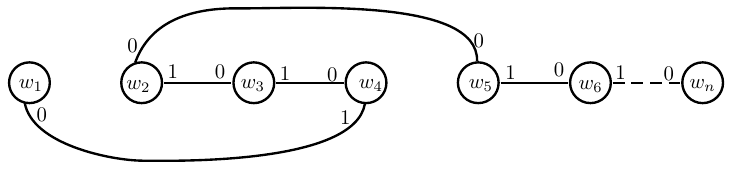}
    \caption{$\mathcal{P}_r'$, the adversary gives $\mathcal{P}_r'$ at round $r$.}
    \label{fig:exp-1-hop-1}
\end{figure}
\begin{proof}
    Since the initial configuration is not dispersed configuration, there are at least two nodes, which are holes, and there is at least one node, which is multinode. Let $V=\{v_1$, $v_2$, \ldots, $v_n\}$ be the set of nodes in $G$, and $v_{n-1}$, $v_n$ be holes, and $v_1$ is a multinode. Consider $n(v)$, which denotes the number of agents at node $v$. At the beginning of round $r\geq 0$, the adversary forms path $\mathcal{P}_r=w_1 \sim w_2\sim \ldots\sim w_n$ of length $n$ based on the following.

    \begin{itemize}
        \item Consider two nodes, \( u \) and \( v \). Let \( a_u \) and \( a_v \) represent the agents with the smallest IDs at nodes \( u \) and \( v \), respectively. Without loss of generality, assume that the ID of agent \( a_u \) is smaller than the ID of agent \( a_v \). If the number of agents at node \( u \) (denoted \( n(u) \)) is greater than the number of agents at node \( v \) (denoted \( n(v) \)), and if \( u \) is represented as \( w_i \) for some \( i \) in the range \( [1, n-1] \), then \( v \) can be represented as \( w_j \) for some \( j > i \). Alternatively, if \( n(u) = n(v) \) and \( u = w_i \) for some \( i \) in the range \( [1, n-1] \), then \( v \) can also be represented as \( w_j \) for some \( j > i \). 
        \item For $i\in [2, n-1]$, node $w_i$ is connected to node $w_{i-1}$ via port 0, and connected to node $w_{i+1}$ via port 1. Node $w_1$ is connected to node $w_2$ via port 0, and node $w_n$ is connected to node $w_{n-1}$ via port 0.
    \end{itemize}
    The adversary provides this graph $\mathcal{P}_r$ in each and every round $r$ unless otherwise mentioned during this proof.

    We show that if there does not exist any round $r\geq 0$ such that the agents are in dispersed configuration by the start of that round, then the exploration problem is impossible to solve by the end of round $r$. In round $r=0$, since node $w_n(=v_n)$ is at least two hops away from agents in $\mathcal{P}_0$, node $w_n(=v_n)$ is unexplored by agents at the end of round $r=0$. Let the agents not be in dispersed configuration till the start of round $r$. Therefore, at the end of round $r$, the way the adversary maintains the path graph $\mathcal{P}_r$ at round $r$, node $w_n(=v_n)$ is at least two hops away from a node with agents in $\mathcal{P}_r$. Therefore, no matter how agents move in $\mathcal{P}_r$, the node $v_n$ is unexplored at the end of round $r$. Therefore, it is sufficient to show that in each round $r>0$, the adversary can restrict agents in $\mathcal{P}_r$ from achieving dispersion such that at the start of round $r+1$, the configuration remains undispersed.

    Consider there is an algorithm $\mathcal{A}$ which solves exploration. If this algorithm does not solve dispersion in some round $r>0$, then the exploration is impossible to solve in round $r+1$ due to the aforementioned reason. Since the adversary is aware of graph $\mathcal{P}_r$ (refer Fig. \ref{fig:exp-1-hop}) and algorithm $\mathcal{A}$, therefore it can pre-compute the outcome of algorithm $\mathcal{A}$ at round $r$. If the pre-computation shows that the agents archive dispersion using algorithm $\mathcal{A}$, then at the beginning of round $r$, it gives the configuration of path $\mathcal{P}_r'$ in place of path $\mathcal{P}_r$. In path $\mathcal{P}_r$ and $\mathcal{P}_r'$, node $w_i$ are the same node. The configuration of the path $\mathcal{P}_r'$ is as follows. 
    
    \begin{itemize}
        \item Node $w_1$ is connected to node $w_4$ via port 0.
        \item Node $w_4$ is connected to node $w_1$ via port 1, and node $w_4$ is connected to node $w_3$ via port 0.
        \item Node $w_2$ is connected to node $w_3$ via port 1, and node $w_2$ is connected to node $w_5$ via port 0.
        \item Node $w_3$ is connected to node $w_2$ via port 0, and node $w_3$ is connected to node $w_4$ via port 1.
        \item Node $w_5$ is connected to node $w_2$ via port 0, and node $w_5$ is connected to node $w_6$ via port 1.
        \item For $i\in [6, n-1]$, node $w_i$ is connected to node $w_{i-1}$ via port 0, and connected to node $w_{i+1}$ via port 1. Node $w_1$ is connected to node $w_2$ via port 0, and node $w_n$ is connected to node $w_{n-1}$ via port 0.
    \end{itemize}
    Therefore, $\mathcal{P}_r'=w_1 \sim w_4\sim w_3\sim w_2\sim\ldots\sim w_{n-1}\sim w_n$ (refer Fig. \ref{fig:exp-1-hop-1}). Note that if the dispersed configuration is achieved at the end of round $r$, then $n(w_1)=2$ and $n(w_i)=1$ for $i\in [2,\,n-2]$. If not, then $n(w_1)\geq3$. In this case, no matter how agents move in $\mathcal{P}_r$, there is a multinode at node $w_1$ or $w_2$ at the end of round $r$. Therefore, there is only one possibility $n(w_1)=2$ and $n(w_i)=1$ for $i\in [2,\,n-2]$ at round $r$. The movement of agent $a_i$ at node $w_3$ as per algorithm $\mathcal{A}$ is the same in $\mathcal{P}_r$ and $\mathcal{P}_r'$ due to the fact the movement of $a_i$ at node $w_3$ depends only on the 1-hop view and its memory. Since, there is no global communication, agent $a_i$ can not understand that it is in $\mathcal{P}_r$ or $\mathcal{P}_r'$. As per algorithm $\mathcal{A}$, the movement of agents at node $w_1$, $w_2$, $w_4$ and $w_5$ can be changed. Since $n\geq 7$, such a configuration is feasible. In this case, the agent at node $w_3$ moves towards node $w_4$ in $\mathcal{P}_r'$. No matter how agents move from node $w_1$, $w_4$, there is a multinode at either node $w_4$ or $w_1$. Therefore, the dispersed configuration is not achieved at the end of round $r$. Since the dispersed configuration is not achieved in round $r> 0$, the node $v_n$ is at least two hops away from agents at the beginning of round $r+1$. This completes the proof.
    \end{proof}

 Below, we provide the other impossibility result considering that 1-hop visibility is not there, but agents are equipped with global communication. In this case, agents can not understand whether the 1-hop view is changed at the beginning of the round. Hence, the agent's decision does not depend on the view.
\begin{theorem}\label{thm:exp-global}
    (For $n\geq 7$) It is impossible to solve the exploration of $n-1$ mobile agents on a 1-Interval Connected dynamic graph when the agents are equipped with global communication and unlimited memory but without 1-hop visibility unless they start in a dispersed configuration.
\end{theorem}

\begin{proof}
    Since the initial configuration is not dispersed configuration, there are at least two nodes, which are holes, and there is at least one node, which is multinode. Let $V=\{v_1$, $v_2$, \ldots, $v_n\}$ be the set of nodes in $G$, and $v_{n-1}$, $v_n$ be holes, and $v_1$ is a multinode. Consider $n(v)$, which denotes the number of agents at node $v$. At the beginning of round $r\geq 0$, the adversary forms path $\mathcal{P}_r=w_1 \sim w_2\sim \ldots\sim w_n$ of length $n$ based on the following. 

    \begin{itemize}
        \item Consider two nodes, \( u \) and \( v \). Let \( a_u \) and \( a_v \) represent the agents with the smallest IDs at nodes \( u \) and \( v \), respectively. Without loss of generality, assume that the ID of agent \( a_u \) is smaller than the ID of agent \( a_v \). If the number of agents at node \( u \) (denoted \( n(u) \)) is greater than the number of agents at node \( v \) (denoted \( n(v) \)), and if \( u \) is represented as \( w_i \) for some \( i \) in the range \( [1, n-1] \), then \( v \) can be represented as \( w_j \) for some \( j > i \). Alternatively, if \( n(u) = n(v) \) and \( u = w_i \) for some \( i \) in the range \( [1, n-1] \), then \( v \) can also be represented as \( w_j \) for some \( j > i \). 
        \item For $i\in [2, n-1]$, node $w_i$ is connected to node $w_{i-1}$ via port 0, and connected to node $w_{i+1}$ via port 1. Node $w_1$ is connected to node $w_2$ via port 0, and node $w_n$ is connected to node $w_{n-1}$ via port 0.
    \end{itemize}
    The adversary provides this graph $\mathcal{P}_r$ in each and every round $r$ unless otherwise mentioned during this proof.

    We show that if there is no round $r\geq 0$ such that the agents are in dispersed configuration by the start of that round, then the exploration problem is impossible to solve by the end of round $r$. In round $r=0$, since node $w_n(=v_n)$ is at least two hops away from agents in $\mathcal{P}_0$, node $w_n(=v_n)$ is unexplored by agents at the end of round $r=0$. Let the agents not be in dispersed configuration till the start of round $r$. Therefore, at the end of round $r$, the way the adversary maintains the path graph $\mathcal{P}_r$ at round $r$, node $v_n$ is at least two hops away from a node with agents in $\mathcal{P}_r$. Therefore, no matter how agents move in $\mathcal{P}_r$, the node $v_n$ is unexplored at the end of round $r$. Therefore, it is sufficient to show that in each round $r>0$, the adversary can restrict agents in $\mathcal{P}_r$ from achieving dispersion such that at the start of round $r+1$, the configuration remains undispersed. 
    
     Consider there is an algorithm $\mathcal{A}$ which solves exploration. If this algorithm does not solve dispersion in some round $r>0$, then the exploration is impossible to solve in round $r+1$ due to the aforementioned reason. Since the adversary is aware of graph $\mathcal{P}_r$ (refer Fig. \ref{fig:exp-1-hop}) and algorithm $\mathcal{A}$, therefore it can pre-compute the outcome of algorithm $\mathcal{A}$ at round $r$. If the pre-computation shows that the agents archive dispersion using algorithm $\mathcal{A}$, then at the beginning of round $r$, it gives the configuration of path $\mathcal{P}_r'$ in place of path $\mathcal{P}_r$. In path $\mathcal{P}_r$ and $\mathcal{P}_r'$, node $w_i$ are the same node. The configuration of the path $\mathcal{P}_r'$ is as follows.

    
    \begin{itemize}
        \item Node $w_1$ is connected to node $w_4$ via port 0.
        \item Node $w_4$ is connected to node $w_1$ via port 1, and node $w_4$ is connected to node $w_3$ via port 0.
        \item Node $w_2$ is connected to node $w_3$ via port 1, and node $w_2$ is connected to node $w_5$ via port 0.
        \item Node $w_3$ is connected to node $w_2$ via port 0, and node $w_3$ is connected to node $w_4$ via port 1.
        \item Node $w_5$ is connected to node $w_2$ via port 0, and node $w_5$ is connected to node $w_6$ via port 1.
        \item For $i\in [6, n-1]$, node $w_i$ is connected to node $w_{i-1}$ via port 0, and connected to node $w_{i+1}$ via port 1. Node $w_1$ is connected to node $w_2$ via port 0, and node $w_n$ is connected to node $w_{n-1}$ via port 0.
    \end{itemize}
    Therefore, $\mathcal{P'}_r=w_1 \sim w_4\sim w_3\sim w_2\sim\ldots\sim w_{n-1}\sim w_n$ (refer Fig. \ref{fig:exp-1-hop}). Note that if the dispersed configuration is achieved at the end of round $r$, then $n(w_1)=2$ and $n(w_i)=1$ for $i\in [2,\,n-2]$. If not, then $n(w_1)\geq3$. In this case, no matter how agents move in $\mathcal{P}_r$, there is a multinode at node $w_1$ or $w_2$ at the end of round $r$. Therefore, there is only one possibility $n(w_1)=2$ and $n(w_i)=1$ for $i\in [2,\,n-2]$ at round $r$. The movement of agent $a_i$ at node $w_i$ as per algorithm $\mathcal{A}$ is the same in $\mathcal{P}_r$ and $\mathcal{P}_r'$ due to the fact the movement of $a_i$ at node $w_i$ depends on the global communication and its memory. Since, there is no 1-hop visibility communication, agent $a_i$ can not distinguish whether it is in $\mathcal{P}_r$ or $\mathcal{P}_r'$ with the help of global communication and its memory. As per algorithm $\mathcal{A}$, the movement of agents at node $w_1$, $w_2$, $w_3$, $w_4$ $\&$ $w_5$ does not change. Since $n\geq 7$, such a configuration is feasible. In this case, the agent at node $w_3$ moves towards node $w_4$ in $\mathcal{P}_r'$. No matter how agents move from node $w_1$, $w_4$, there is a multinode at node $w_4$. Therefore, the dispersed configuration is not achieved at the end of round $r$. Since the dispersed configuration is not achieved in round $r> 0$, the node $v_n$ is at least two hops away from agents at the beginning of round $r+1$. This completes the proof. 
\end{proof}

Based on Theorem \ref{thm:exp-1-hop}, \ref{thm:exp-global}, we have the following observations.
\begin{observation}
    To solve the exploration problem in 1-Interval Connected graphs, a team of n-1 agents requires 1-hop visibility and global communication.
\end{observation}

\begin{observation}
    It is impossible to solve the exploration problem with $n-1$ agents in $T$-Path Connected graphs or the Connectivity Time graphs using either only 1-hop visibility or only global communication due to Observation \ref{obs:connectivity} unless they are dispersed.
\end{observation}

Theorem \ref{thm:exp-1-hop}, \ref{thm:exp-global} are valid when agents are not in the dispersed configuration. If agents are in the dispersed configuration, then what is the necessary condition to solve the exploration in 1-Interval Connected graphs? We answer this question as follows.
\begin{theorem}\label{thm:disp-Exp}
    (For $n\geq 3$) Let $n-1$ agents be dispersed initially. It is impossible to solve the exploration of $n-1$ mobile agents on a 1-Interval Connected dynamic graph when they are equipped with global communication and unlimited memory, know the parameter $k,n$, and are aware of dispersed configuration but without 1-hop visibility.
\end{theorem}

\begin{figure}
    \centering
    \includegraphics[width=1\linewidth]{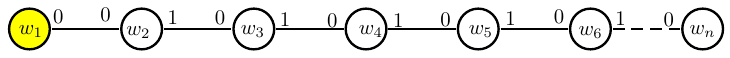}
    \caption{$\mathcal{P}_0$, the adversary supposed to give $\mathcal{P}_0$ at round $r=0$. Node $w_1$ is a hole.}
    \label{fig:EXP_11}
\end{figure}

\begin{figure}
    \centering
    \includegraphics[width=1\linewidth]{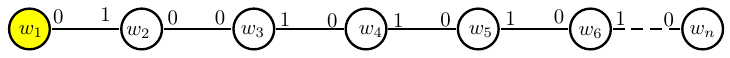}
    \caption{$\mathcal{P}_0'$, the adversary gives $\mathcal{P}_0'$ at round $r=0$. Node $w_1$ is a hole.}
    \label{fig:EXP_12}
\end{figure}

\begin{proof}
    Let $v_1$, $v_2$, \ldots, $v_n$ be nodes, and $n-1$ agents are dispersed initially. Without loss of generality, let $v_1$ be a hole initially, and $v_i$ contains at least 1 agent $\forall\; i>1$. At the beginning of round $r\geq 0$, if agents are in the dispersed configuration, then the adversary forms path $\mathcal{P}_r=w_1 \sim w_2\sim \ldots\sim w_n$ of length $n$ based on the following.

    \begin{itemize}
        \item Let agent $a_i$ be at node $w_i$, $i\geq 2$, and node $w_1$ be a hole.
        \item For $i\in [2, n-1]$, node $w_i$ is connected to node $w_{i-1}$ via port 0, and connected to node $w_{i+1}$ via port 1. Node $w_1$ is connected to node $w_2$ via port 0, and node $w_n$ is connected to node $w_{n-1}$ via port 0.
    \end{itemize}
    
    The adversary provides this graph $\mathcal{P}_r$ in each and every round $r$ unless otherwise mentioned during this proof.

    According to the construction described above, a path $\mathcal{P}_0=w_1(=v_1)\sim w_2\sim \ldots w_n$ is expected to be provided at round $0$ (see Fig. \ref{fig:EXP_11}). If there is any algorithm $\mathcal{A}$ which solves exploration, then the adversary can pre-compute the movement of agents in $\mathcal{P}_0$ as it is aware of algorithm $\mathcal{A}$ and configuration $\mathcal{P}_0$. Since agents are equipped with global communication but no 1-hop visibility, their movement decision is based on global communication and existing memory. As per pre-computation, if agent $a_2$ at node $w_2$ does not move to node $w_1$ in the path $\mathcal{P}_0$ at round $r=0$, the adversary gives the configuration $\mathcal{P}_0$ at the beginning of round $r=0$. If agent $a_2$ at node $w_2$ moves to node $w_1$ in the path $\mathcal{P}_0$ at round $0$, the adversary gives configuration $\mathcal{P}_0'$ (refer Fig. \ref{fig:EXP_12}) instead of $\mathcal{P}_0$ at the beginning of round 0 as follows: via port 1, node $w_2$ is connected with node $w_1$, and via port 0, node $w_2$ is connected with node $w_3$. Since agents lack 1-hop visibility, the decision of agent $a_2$ remains the same as $\mathcal{P}_0$ in $\mathcal{P}_0'$. Therefore, at the end of round $r=0$, the node $v_1$ remains unvisited. If at the end of round $r\geq 0$, the movement of agents in this configuration leads to non-dispersed configuration, then we can not solve exploration for any round $r>0$ using Theorem \ref{thm:exp-global}. If, at the end of round $r\geq 0$, the configuration remains dispersed, we can use the same idea of $r=0$. In every round $r\geq 0$, the node $v_1$ remains unexplored. Therefore, if $n-1$ agents are dispersed initially, it is impossible to solve the exploration of $n-1$ mobile agents on a 1-Interval Connected dynamic graph when equipped with global communication and unlimited memory but without 1-hop visibility.
\end{proof}

Based on Theorem \ref{thm:disp-Exp}, we have the following observation.

\begin{observation}\label{obs:necessary-exp-1-hop}
    If agents are in a dispersed configuration in 1-Interval Connected graphs, then without 1-hop visibility, it is impossible to solve the exploration problem.
\end{observation}

We ask the following question.

\begin{question}
    Whether global communication is necessary to solve exploration in 1-Interval Connected graphs when agents are aware of a dispersed configuration.
\end{question}
\begin{answer}
    The answer is negative for the following reason: consider that the agents are in a dispersed configuration and equipped with 1-hop visibility. Using a straightforward algorithm, the agents can complete the exploration with termination in a single round. In this scenario, if a team of \( n-1 \) agents is in a dispersed configuration, there will be one node (let's call it \( v \)) that remains unexplored. Since \(\mathcal{G}_0\) is connected due to 1-Interval Connectivity, at least one agent must be located in a neighbouring node of \( v \) within \(\mathcal{G}_0\). By the end of round 0, if the agents move to the neighbouring node, which is a hole, they will realize that the exploration is complete and can terminate. Therefore, exploration can be resolved in one round if the agents are aware of the dispersed configuration and have 1-hop visibility.
\end{answer}

\begin{remark}\label{rk:T-Path}
    It is impossible to solve the exploration problem with $n-1$ agents in $T$-Path Connected graphs or the Connectivity Time graphs using either only 1-hop visibility or only global communication due to Theorem \ref{thm:exp-1-hop}, \ref{thm:exp-global} and Observation \ref{obs:connectivity} unless they are dispersed.
\end{remark}





\begin{figure}
    \centering
    \includegraphics[width=0.5\linewidth]{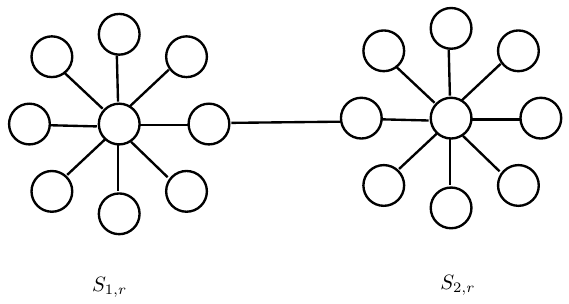}
    \caption{$\mathcal{G}_r$ at round $r\geq0$}
    \label{fig:exp_1_int_time}
\end{figure}

The following Theorem is in terms of time lower bound in 1-Interval Connected Dynamic Graphs, which is as follows.
\begin{theorem}\label{thm:time_lower_exp}
    Any algorithm solving exploration problem in 1-Interval Connected Dynamic graph of $n$ nodes requires $\Omega(n)$ rounds even if $\hat{D}=O(1)$. Moreover, this result holds if the agents have infinite memory, are equipped with global communication, have full visibility and know all of $k$, $n$, $T$.
\end{theorem}
\begin{proof}
    Let $V=\{v_1$, $v_2$, \ldots, $v_n\}$ be the set of nodes in $G$, and agents are co-located at node $v_1$. We demonstrate that it is possible to construct $\mathcal{G}_r$ at each round $r$ such that, at most, one new node is visited by the agents. This shows that there exists a dynamic graph on which the agents require a minimum of \( n \) rounds to visit \( n \) new nodes. Consequently, this implies that the exploration process requires \(\Omega(n)\) rounds. Though a path is a trivial configuration, we provide a dynamic graph $G$ with $\hat{D}=O(1)$. The construction of $\mathcal{G}_r$ is as follows. 

    Initially, all agents are co-located at node $v_1$. At each round, $r\geq 0$, let $S_{1,r}$ be a set of node(s) which are visited by agents at least once, and $S_{2,r}$ be a set of node(s) which are not visited yet by agents. At the beginning of round $r$, the adversary forms two-star graphs $G_1$ and $G_2$ from sets $S_{1,r}$ and $S_{2,r}$ and connects them by one edge (refer to Fig. \ref{fig:exp_1_int_time}). In this case, in each round $r$, the graph holds 1-Interval connectivity due to the fact that $\mathcal{G}_r$ is connected in each round. Initially, $S_{1,0}=\{v_1\}$, and $S_{2,0}=V-\{v_1\}$. In $\mathcal{G}_r$, agents can visit at most one unvisited node no matter how agents move. To visit $n$ nodes in this dynamic graph sequence, agents need at least $n$ rounds due to the fact that they can visit at most one unvisited node per round. Therefore, to solve the exploration problem in the 1-Interval Connected Dynamic graph, agents need $\Omega(n)$ rounds.
\end{proof}

The following theorem is for time lower bound in $T$-Path Connected graphs.
\begin{theorem}\label{thm:time_lower_exp1}
    Any algorithm solving exploration problem in $T$-Path Connected graph of $n$ nodes requires $\Omega(n\cdot T)$ rounds even if $\hat{D}=O(1)$. Moreover, this result holds if the agents have infinite memory, are equipped with global communication, have full visibility and know all of $k$, $n$, $T$.
\end{theorem}

\begin{proof}
    The proof is similar to Theorem \ref{thm:lower_bound_T}. For the sake of completeness, we recall the idea. Consider that $n-1$ agents are co-located at some node. In each $T$ round, it connects to a new unexplored node for a $1$ round and disconnects the remaining unexplored nodes for the next $T$ rounds. In this way, to visit $n-2$ nodes, the agents need at least $(n-2)\cdot T$ rounds. Therefore, any algorithm solving exploration problem of any $T$-Path Connected graph of $n$ nodes requires $\Omega(n\cdot T)$ rounds.
\end{proof}

\begin{figure}[ht]
    \centering
    \begin{minipage}{0.48\linewidth}
        \centering
        \includegraphics[width=0.75\linewidth]{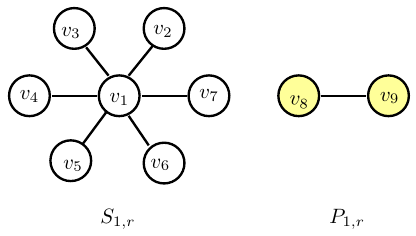}
        \caption{Dynamic graph at the beginning \\of round $r\in[0,T-2]$.}
        \label{fig:r=0}
    \end{minipage}
    \begin{minipage}{0.48\linewidth}
        \centering
        \includegraphics[width=0.93\linewidth]{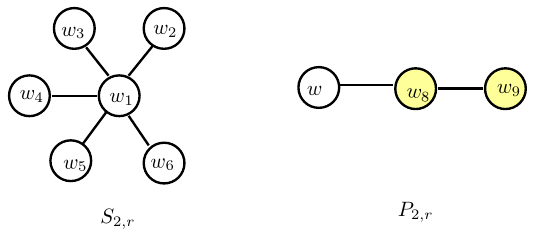}
        \caption{Dynamic graph at the beginning \\of round $r=iT-1$.}
        \label{fig:r=iT-1}
    \end{minipage}
\end{figure}

\begin{figure}[ht]
    \centering
    \begin{minipage}{0.48\linewidth}
        \centering
        \includegraphics[width=0.95\linewidth]{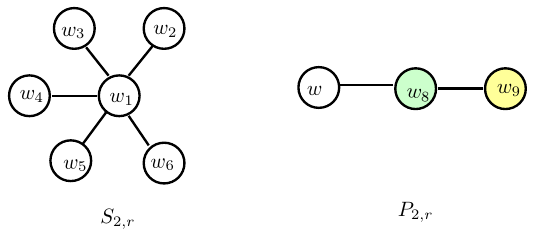}
        \caption{Dynamic graph at the end of round \\$r=iT-1$.}
        \label{fig:r=iT-1_end}
    \end{minipage}
    \begin{minipage}{0.48\linewidth}
        \centering
        \includegraphics[width=0.8\linewidth]{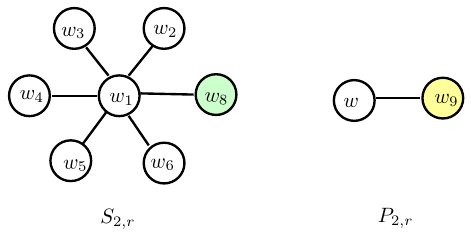}
        \caption{Dynamic graph at the beginning\\ of round $r=iT$ if $n(w)=0$.}
        \label{fig:r=iT,1}
    \end{minipage}
\end{figure}

\begin{figure}[ht]
    \centering
    \begin{minipage}{0.48\linewidth}
        \centering
        \includegraphics[width=0.8\linewidth]{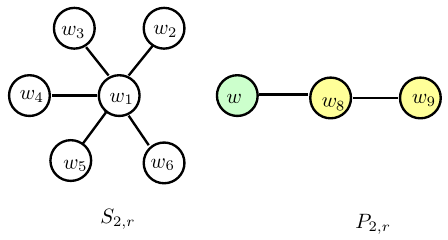}
        \caption{Dynamic graph at the end of round \\$r=iT-1$ if $n(w)=1$.}
        \label{fig:r=iT-1,end_1}
    \end{minipage}
    \begin{minipage}{0.48\linewidth}
        \centering
        \includegraphics[width=0.8\linewidth]{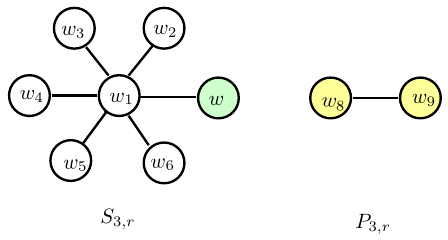}
        \caption{Dynamic graph at the beginning \\of round $r=iT$ if $n(w)=1$.}
        \label{fig:r=iT,11}
    \end{minipage}
\end{figure}

\begin{theorem}\label{thm:connectivity_exp}
   ($n\geq 6$) If the initial configuration contains at least two holes, then a group of \( k \leq n \) agents cannot solve the exploration problem in dynamic graphs that maintain the Connectivity Time property. This impossibility holds even if agents have infinite memory, full visibility, global communication, and know the parameters $k$, $n$, $T$.
\end{theorem}
\begin{proof}
Let \( v_1, v_2, \ldots, v_n \) be the nodes in a dynamic graph \( G \), and $n(v_i)$ denote the number of agents at node $v_i$. Without loss of generality, let nodes \( v_{n-1} \) and \( v_n \) be holes. There are \( n \) agents located at nodes \( v_1, \ldots, v_{n-2} \), meaning that the agents may either be co-located or scattered among these nodes. The adversary maintains the dynamic graph as follows.

\begin{itemize}
   \item $\bm{r\in[0, \,T-2]:}$ At every round $r$, it forms $\mathcal{G}_r$ as follows. The graph $\mathcal{G}_r$ contains two connected components (for $n=9$, refer to Fig. \ref{fig:r=0}): (i) a star graph $S_{1,r}$ of nodes $v_1$, $v_2$,\ldots, $v_{n-2}$, and (ii) path $P_{1,r}=v_{n-1}\sim v_{n}$. 
   
   \item $\bm{r=iT-1}, \textbf{where } \bm{i\in \mathbb{N}:}$ At the end of round $iT-2$, there is a star graph $\mathcal{S}$ (contains $n-2$ nodes) and path $\mathcal{P}$ (contains two nodes) in $\mathcal{G}_{iT-2}$. Let $w_1$, $w_2$, \ldots, $w_{n-2}$ be nodes in $\mathcal{S}$, and $w_{n-1}$, $w_{n}(=v_n)$ be nodes in $\mathcal{P}$. Since $n-2$ nodes are in star graph $\mathcal{S}$, and $n\geq 6$, therefore there is a node $w$ in $\mathcal{S}$ such that $n(w)\leq 1$. At the beginning of round $r=iT-1$, the adversary forms a new dynamic graph from $\mathcal{G}_{r}$. The graph $\mathcal{G}_{r}$ contains two connected components (for $n=9$ refer to Fig. \ref{fig:r=iT-1}): (i) a star graph $S_{2,r}$ from a set of nodes $Q$, where $Q=\{w_1, w_2, \ldots, w_{n-2}\}-\{w\}$ and (ii) path $P_{2,r}=w \sim w_{n-1}\sim w_{n}$.
   
    \item $\bm{r \in [iT, \, (i+1) T-2]}$, \textbf{where} $\bm{i \in \mathbb{N}:}$ At the end of round $iT-1$, there is a star graph $S_{2,iT-1}$ (contains $n-3$ nodes) and path $\mathcal{P}_{2,iT-1}$ (contains three nodes) in $\mathcal{G}_{iT-1}$. Let $w_1$, $w_2$, \ldots, $w_{n-3}$ be nodes in $S_{2, iT-1}$, and $w$, $w_{n-1}$, $w_{n}(=v_n)$ be nodes in $P_{2, iT-1}$. As per construction of $\mathcal{G}_{iT-1}$, $n(w_1)\leq 1$. At the beginning of round $r=iT-1$, if $n(w)=0$, then adversary forms $\mathcal{G}_{iT}$ as follows. The graph $\mathcal{G}_{iT}$ contains two connected components: (i) a star graph $S_{3,r}$ from a set of nodes $Q$, where $Q=\{w_1, w_2, \ldots, w_{n-2}\}\cup\{w\}$ and (ii) path $P_{2,r}=w_{n-1}\sim w_{n}$. 

    At the beginning of round $r=iT-1$, if $n(w)=1$, then adversary forms $\mathcal{G}_{iT}$ as follows. If at the end of round $iT-1$, the agent at node $w$ moves from node $w$ to node $w_{n-1}$ (for $n=9$ refer to Fig. \ref{fig:r=iT-1_end}), the graph $\mathcal{G}_{iT}$ contains two connected components (for $n=9$ refer to Fig. \ref{fig:r=iT,1}): (i) a star graph $S_{3,r}$ from a set of nodes $Q$, where $Q=\{w_1, w_2, \ldots, w_{n-2}\}\cup \{w_{n-1}\}$ and (ii) path $P_{3,r}=w \sim w_{n}$. If at the end of round $iT-1$, the agent at node $w$ stays at its position (for $n=9$ refer to Fig. \ref{fig:r=iT-1,end_1}), then the graph $\mathcal{G}_{iT}$ contains two connected components (for $n=9$ refer to Fig. \ref{fig:r=iT,11}): (i) a star graph $S_{3,r}$ from a set of nodes $Q$, where $Q=\{w_1, w_2, \ldots, w_{n-2}\}\cup \{w\}$ and (ii) path $P_{3,r}=w_{n-1} \sim w_{n}$. At every round $r>iT$, the adversary maintains $\mathcal{G}_r=\mathcal{G}_{iT}$.
  \end{itemize} 
This dynamic setting satisfies the connectivity time property due to the following reasons. For $r\geq 0$ and $r\neq iT-1$ for some $i\in \mathbb{N}$, let $\mathcal{G}_r$, $\mathcal{G}_{r+1}$, \ldots, $\mathcal{G}_{r+T-1}$ be consecutive $r$ sequence of graphs, where $\mathcal{G}_j=(V, E(j))$ for $j\in[r,r+T-1]$. Suppose the above dynamic graph $G$ does not satisfy the Connectivity Time property on interval $[r, r+T-1]$, i.e., $G_{r, T}:=(V, \cup_{j=r}^{r+T-1} E(j))$ is not connected. Note that there exists a round $r'$ between $r$ and $r+T-1$ such that $r'=iT-1$, for some $i\in \mathbb{N}$. In each round between $r$ and $r'-1$, there are two connected components in $\mathcal{G}_r$. Without loss of generality, let $\mathcal{S}$ be a star graph of $n-2$ nodes, and $\mathcal{P}$ be length 1 path at the start of round $r'-1$. Let $u_1$, $u_2$, \ldots, $u_{n-2}$ be nodes in $\mathcal{S}$, and $u_{n-1}$, $u_{n}$ be nodes in $\mathcal{P}$. The adversary changes the configuration at round $r'$ as follows. At the beginning of round $r'$, it forms a star graph of $n-3$ nodes from nodes $u_1$, $u_2$, \ldots, $u_{n-2}$, and it select a node $u$ from nodes $u_1$, $u_2$, \ldots, $u_{n-2}$, and from a path of two length from nodes $u$, $u_{n-1}$, $u_n$. In this case, the graph $G_{r, T}$ is a connected component. This shows our assumption is wrong. Similarly, we can show for $r=iT-1$. Therefore, this dynamic setting satisfies the Connectivity Time property.

To prove our theorem, it is sufficient to show that the node $v_n$ is not accessible to the agents in each round $r\geq 0$. If $r\in[0, \, T-2]$, node $v_n$ is not accessible to the agents because node $v_n$ is in connected component $\mathcal{P}_{1, r}=v_{n-1}\sim v_n$, and nodes $v_{n-1}$ and $v_{n}$ are holes. 

If $r=T-1$, then the node $v_n$ is not accessible by agents due to the following reason: at the beginning of round $r=T-1$, there can be one agent at node $w$. No matter how the agent moves in round $r$, it can not access node $v_n$ as $n\geq 6$.

If $r\in [T, 2T-2]$, then the node $v_n$ is not accessible to the agents because node $v_n$ is in the connected component of holes. The same idea can be extended for $r\geq 2T-1$. This completes the proof.
\end{proof}

\section{Exploration in Dynamic Graphs}
\subsection{1-Interval Connected Dynamic Graph Exploration}\label{sec:EXP-1-Int}
In this section, we present an algorithm which solves exploration in 1-Interval Connected dynamic graphs. As per Theorem \ref{thm:imp_Exp_1-Interval}, \ref{thm:exp-1-hop} and \ref{thm:exp-global}, the $n-1$ agents need global communication and 1-hop visibility to solve the exploration. 

\subsubsection{The Algorithm}
The algorithm is based on $\mathcal{DISP}$ (recall from Section \ref{sec:pre}). Agents maintain a parameter which is as follows. 
\begin{itemize}
    \item $a_i.count:$  It represents the number of agents present in the node including $a_i$. If $a_i.count=1$, then the current node where $a_i$ resides is not a multinode; else, if $a_i.count>1$, then the current node is a multinode. 
\end{itemize} 
Refer to Algorithm \ref{algorithm_EXP} for the pseudocode.

\begin{algorithm}
    \caption{Exploration with Termination}
    \label{algorithm_EXP}
    \While{True}
    {
        agent $a_i$ broadcasts 1-hop neighbours information and $a_i.count$ \\
        \If{$a_i.count>1$ or $a_i$ receives $a_j.count>1$ }
        {
            execute $\mathcal{DISP}$ at this round
        }
        \Else
        {
            \If{agent $a_i$ finds a hole in its 1-hop neighbour}
            {
                moves to the hole and terminates.
            }
            \ElseIf{ agent $a_i$ does not find a hole in its 1-hop neighbour}
            {
                terminates
            }
        } 
    }
\end{algorithm}
\subsubsection{Correctness and Analysis of the Algorithm}
The analysis of the algorithm is as follows. 
\begin{theorem}\label{thm:exp-1-int-sol}
    Our algorithm solves exploration with $n-1$ agents in $\Theta(n)$ rounds using $O(\log n)$ bits of memory per agent in the synchronous setting with global communication and 1-hop visibility.
\end{theorem}
\begin{proof}
    In 1-Interval Connected dynamic graphs, $\mathcal{G}_r$ is connected at each round $r$. Due to that, if there is any multinode at round $r$, all agents get the information of the multinode at round $r$, and agents run $\mathcal{DISP}$ at round $r$. Due to Theorem \ref{thm:main_result}, the agents achieve the dispersion in $\Theta(n)$ rounds and also understand that dispersion has been achieved with the help of global communication (as there will be no multinode). Since $n-1$ agents are in dispersed configuration, there will be a node $v$ which agents might not be able to visit. Since $\mathcal{G}$ is connected in each round $r$, and $n-1$ agents are present, the node $v$ is connected to a node $u$ which has an agent $a_i$. As per our algorithm, when agents do not hear a multinode, the agents move to the node which is a hole. In this case, $a_i$ moves to node $v$. This completes the exploration as each node is visited by agents at least once. Our algorithm takes $O(n)$ rounds and $O(\log n)$ memory per agent due to Theorem \ref{thm:main_result}. Due to Theorem \ref{thm:time_lower_exp}, our algorithm solves the exploration in $\Theta(n)$ rounds using $O(\log n)$ bits of memory per agent in the synchronous setting with global communication and 1-hop visibility. This completes the proof.
\end{proof}

\subsection{$T$-Path Connected Graph Exploration}\label{sec:EXP-T-Int}
In this section, we present an algorithm to solve exploration in $T$-Path Connected Graphs. Due to Remark \ref{rk:T-Path Connected} and \ref{rk:T-Path}, $n-1$ agents with global communication and 1-hop visibility are required to solve the exploration problem. We present an algorithm when agents are unaware of $T$. Since the agents are unaware of $n,k, T$, the agents can not understand whether the dispersion has been achieved. We have to modify the dispersion algorithm in such a way the exploration is achieved.

\begin{algorithm}
    \caption{Perpetual Exploration}
    \label{algorithm_EXP1}
    \While{True}
    {
agent $a_i$ broadcasts 1-hop neighbours information, $a_i.ID$ and $a_i.count$ \\
        \If{$a_i.count>1$}
        {
            \If{$a_i$ receives $a_j.count>1$ from some agent $a_j$}
            {
                \If{$a_i$ has a hole in 1-hop or it receives a hole from some agent $a_k$}
                {
                    \If{$a_i$ is minimum ID agent among $a_j$}
                    {
                        it computes sliding path as per $\mathcal{DISP}$
                    }
                    \ElseIf{$a_i$ is not minimum ID agent among $a_j$}
                    {
                        $a_i$ computes a sliding path based on minimum ID agent among $a_j$. It moves if $a_i$ is on the sliding path as per $DISP$ of $a_j$. Otherwise, it waits at its position.
                    }
                }
                \ElseIf{$a_i$ has no hole in 1-hop \& it does not receive a hole from some agent $a_k$}
                {
                    it stays at its position
                }
            }
            \If{$a_i$ does not receive $a_j.count>1$ from some agent $a_j$}
            {
                \If{$a_i$ has a hole in 1-hop or it receives a hole from some agent $a_k$}
                {
                    it computes sliding path as per $\mathcal{DISP}$
                }
                \ElseIf{$a_i$ has no hole in 1-hop $\&$ it does not receive a hole from some agent $a_k$}
                {
                    it stays at its position
                }
            }
        }
        \ElseIf{$a_i.count=1$}
        {
            \If{it receives $a_j.count>1$ from some agent $a_j$}
            {
                \If{$a_i$ has a hole in 1-hop or it receives a hole from some agent $a_k$}
                {
                    $a_i$ computes a sliding path based on minimum ID agent among $a_j$. It moves if $a_i$ is on the sliding path as per $DISP$ of $a_j$. Otherwise, it waits at its position.
                }
                \ElseIf{$a_i$ has no hole in 1-hop \& it does not receive a hole from some agent $a_k$}
                {
                    it stays at its position
                }
            }
            \ElseIf{it does not receives $a_j.count>1$ from some agent $a_j$}
            {
                \If{$a_i$ is minimum ID with a hole in its neighbour}
                {
                    $a_i$ moves through the minimum port, which leads to a hole
                }
                \Else
                {
                    wait at the current node
                }
            }    
        }
}
\end{algorithm}
\subsubsection{The Algorithm}
Here, we modify the algorithm which we mentioned in Section \ref{sec:algorithm}. A detailed description of the algorithm is as follows. Agents maintain a parameter which is as follows. 
\begin{itemize}
    \item $a_i.count:$  It represents the number of agents present in the node, including $a_i$. If $a_i.count=1$, then the current node where $a_i$ resides is not a multinode; else, if $a_i.count>1$, then the current node is a multinode. 
\end{itemize}

At round $r$, $\mathcal{G}_r$ may be disconnected. In each round $r$, agent $a_i$ broadcast the information of its 1-hop along with $a_i.count$. Based on this information, agent $a_i$ does the following.

 \begin{enumerate}
        \item If agent $a_i.count>1$ and it receives $a_j.count>1$ from at least one other agent $a_j$, it works as mentioned below. 
        \begin{itemize}
            \item If $a_i$ receives a hole information, it performs the following steps.
            \begin{itemize}
                \item If $a_i$ is minimum ID among $a_j$s, then it runs $\mathcal{DISP}$.
                \item If $a_i$ is not minimum ID among $a_j$s, then it moves as follows. Let $a_j$ be the minimum ID agent. If $a_i$ is on the sliding path as per $\mathcal{DISP}$ of $a_j$, then it moves. Otherwise, it stays at its position.
            \end{itemize}
            \item If $a_i$ does not receive hole information, it stays at its position.
        \end{itemize}
        \item If agent $a_i.count>1$ does not receive $a_j.count>1$ from some agent $a_j$, then it makes the following decision.
        \begin{itemize}
            \item If $a_i$ receives a hole information, then it runs $\mathcal{DISP}$.
            \item If $a_i$ does not receive hole information, it stays at its position.
        \end{itemize}
        \item If agent $a_i.count=1$ and it receives $a_j.count>1$ from some agent $a_j$, then it makes the following decision.
        \begin{itemize}
            \item Let $a_k$ be minimum ID agent among $a_j$s. If $a_i$ receives a hole information or there is a hole in 1-hop of $a_i$, and $a_i$ is on the sliding path based on $\mathcal{DISP}$ of $a_j$, then it moves on the sliding path.
            \item If $a_i$ does not receive the hole information and there is no hole in its neighbours, it stays at its position.
        \end{itemize}

        \item If agent $a_i.count=1$ and it does not receive $a_j.count>1$ from some agent $a_j$, then it makes the following decision.
        \begin{itemize}
            \item If $a_i$ is a minimum ID agent with a hole in its neighbour, then it moves through the minimum port, which leads to the hole. If not, then it stays at its position.  
            \item If the neighbour of $a_i$ is not a hole and also does not receive information about the hole, then it stays at its position.
        \end{itemize}
    \end{enumerate}

    Refer to Algorithm \ref{algorithm_EXP1} for the pseudocode.
\begin{figure}[h]
\centering
\begin{subfigure}[b]{0.32\textwidth}
\centering
  \scalebox{0.8}{
\begin{tikzpicture}
\node[noeud] (0) at (0,1) {$2$};
\node[noeud,fill=red] (1) at (1,1) {};
\node[noeud] (2) at (2,1) {};
\node[noeud] (3) at (3,1) {};
\draw (0)to(1);
            \draw (2)to(3);
\node (empty) at (0,0) {};
            \node (v_1) at (0,1.5) {$v_1$};
            \node (v_2) at (1,1.5) {$v_2$};
            \node (v_3) at (2,1.5) {$v_3$};
            \node (v_4) at (3,1.5) {$v_4$};
\end{tikzpicture}
  }
\caption{$\mathcal{G}_r$, where $r (\text{mod}\; 6)=0$.}
\label{Fig:Case-1}
\end{subfigure}\hfill
\begin{subfigure}[b]{0.32\textwidth}
\centering
  \scalebox{0.8}{
\begin{tikzpicture}
        \node[noeud] (0) at (-1,1) {$2$};
\node [rotate=270,inner sep=0pt] (0') at (0.6,0.5){\tikz\draw[-triangle 90](1,-0.5);};
\node[noeud,fill=red] (1) at (0,1) {};
\node[noeud] (2) at (1,1) {};
\node[noeud] (3) at (2,1) {};
\draw (1) to (2)to (3);
            \draw[bend right] (1)to(0.5,0.5)to(2);
\node (empty) at (0,0) {};
            \node (v_1) at (-1,1.5) {$v_1$};
            \node (v_2) at (0,1.5) {$v_2$};
            \node (v_3) at (1,1.5) {$v_3$};
            \node (v_4) at (2,1.5) {$v_4$};

  \end{tikzpicture}}
\caption{$\mathcal{G}_r$, where $r (\text{mod}\;6)=1$.}
\label{Fig:Case-2}
\end{subfigure}\hfill
\begin{subfigure}[b]{0.32\textwidth}
\centering
  \scalebox{0.8}{
\begin{tikzpicture}
     \node[noeud] (0) at (0,1) {$2$};
\node[noeud,fill=red] (1) at (1,1) {};
\node[noeud] (2) at (2,1) {};
\node[noeud] (3) at (3,1) {};
\draw (0)to(1);
            \draw (2)to(3);
\node (empty) at (0,0) {};
            \node (v_1) at (0,1.5) {$v_1$};
            \node (v_2) at (1,1.5) {$v_3$};
            \node (v_3) at (2,1.5) {$v_4$};
            \node (v_4) at (3,1.5) {$v_2$};
\end{tikzpicture}
  }
\caption{$\mathcal{G}_r$, where $r (\text{mod}\, 6)=2$.}
\label{Fig:Case-3}
\end{subfigure}\hfill
    \begin{subfigure}[b]{0.32\textwidth}
\centering
  \scalebox{0.8}{
\begin{tikzpicture}
        \node[noeud] (0) at (-1,1) {$2$};
\node [rotate=270,inner sep=0pt] (0') at (0.6,0.5){\tikz\draw[-triangle 90](1,-0.5);};
\node[noeud,fill=red] (1) at (0,1) {};
\node[noeud] (2) at (1,1) {};
\node[noeud] (3) at (2,1) {};
\draw (1) to (2)to (3);
            \draw[bend right] (1)to(0.5,0.5)to(2);
            \node (empty) at (0,2.5) {};
\node (empty) at (0,0) {};
            \node (v_1) at (-1,1.5) {$v_1$};
            \node (v_2) at (0,1.5) {$v_3$};
            \node (v_3) at (1,1.5) {$v_4$};
            \node (v_4) at (2,1.5) {$v_2$};

\end{tikzpicture}}
\caption{$\mathcal{G}_r$, where $r (\text{mod}\; 6)=3$.}
\label{Fig:Case-4}
\end{subfigure}\hfill
\begin{subfigure}[b]{0.32\textwidth}
\centering
  \scalebox{0.8}{
\begin{tikzpicture}
\node[noeud] (0) at (0,1) {$2$};
\node[noeud,fill=red] (1) at (1,1) {};
\node[noeud] (2) at (2,1) {};
\node[noeud] (3) at (3,1) {};
\draw (0)to(1);
            \draw (2)to(3);
\node (empty) at (0,0) {};
            \node (v_1) at (0,1.5) {$v_1$};
            \node (v_2) at (1,1.5) {$v_4$};
            \node (v_3) at (2,1.5) {$v_2$};
            \node (v_4) at (3,1.5) {$v_3$};
\end{tikzpicture}
  }
\caption{$\mathcal{G}_r$, where $r (\text{mod}\; 6)=4$.}
\label{Fig:Case-5}
\end{subfigure}\hfill
\begin{subfigure}[b]{0.32\textwidth}
\centering
  \scalebox{0.8}{
\begin{tikzpicture}
        \node[noeud] (0) at (-1,1) {$2$};
\node [rotate=270,inner sep=0pt] (0') at (0.6,0.5){\tikz\draw[-triangle 90](1,-0.5);};
\node[noeud,fill=red] (1) at (0,1) {};
\node[noeud] (2) at (1,1) {};
\node[noeud] (3) at (2,1) {};
\draw (1) to (2)to (3);
            \draw[bend right] (1)to(0.5,0.5)to(2);
\node (empty) at (0,0) {};
            \node (v_1) at (-1,1.5) {$v_1$};
            \node (v_2) at (0,1.5) {$v_4$};
            \node (v_3) at (1,1.5) {$v_2$};
            \node (v_4) at (2,1.5) {$v_3$};

  \end{tikzpicture}}
\caption{$\mathcal{G}_r$, where $r (\text{mod}\; 6)=5$.}
\label{Fig:Case-6}
\end{subfigure}
\caption{This is the example of our algorithm for $T = 6$ and $n = 4$ where the dispersion is not achieved, but the
exploration has been achieved as per our algorithm. In this figure, the red colour node has one agent, node $v_1$ has
two agents, and other nodes are holes.}
\label{fig-}
\end{figure}
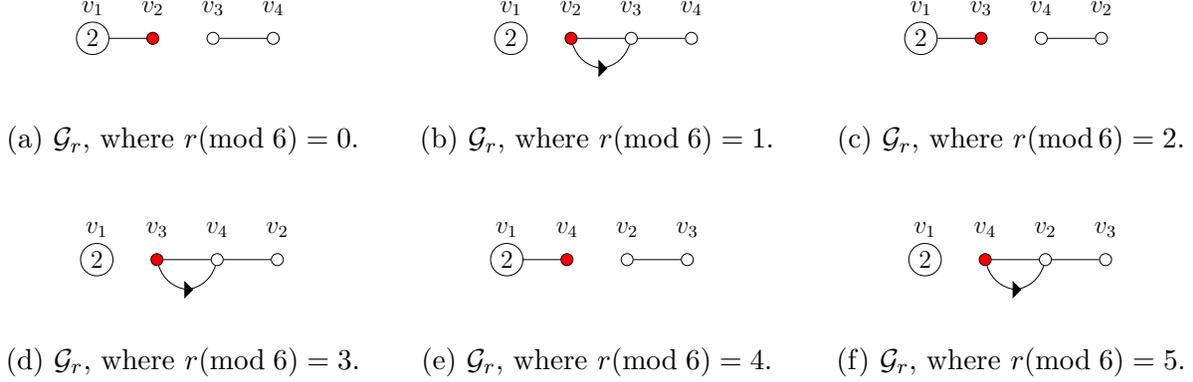

\begin{remark}
       The adversary can form a sequence of graphs in such a way the dispersion is not solved using Algorithm \ref{algorithm_EXP1}. Let $T=6$, and $n=4$. At round $r=0$, node $v_1$ is a multinode, $v_2$ is a node with an agent, and $v_3$, $v_4$ are holes. At round $r=1$ (refer Fig. \ref{Fig:Case-1}), a path of length 1 using nodes $v_1$ and $v_2$ is formed, and a path of length 1 using nodes $v_3$ and $v_4$ is formed. As per Algorithm \ref{algorithm_EXP1}, agents stay in their position in $r=0$. At round $r=1$ (refer Fig. \ref{Fig:Case-2}), node $v_1$ becomes an isolated node and a path $P_1=v_2\sim v_3\sim v_4$ is formed. As per Algorithm \ref{algorithm_EXP1} at round $r=1$, the agent moves to node $v_3$ from node $v_2$. At round $r=2$ (refer Fig. \ref{Fig:Case-3}), a path of length 1 using nodes $v_1$ and $v_3$ is formed, and a path of length 1 using nodes $v_4$ and $v_2$ is formed. As per Algorithm \ref{algorithm_EXP1}, agents stay in their position in $r=2$. At round $r=3$ (refer Fig. \ref{Fig:Case-3}), node $v_1$ becomes an isolated node and a path $P_2=$ $v_3 \sim v_4\sim v_2$ is formed. As per Algorithm \ref{algorithm_EXP1} at round $r=3$, the agent moves to node $v_4$ from node $v_3$. At round $r=4$ (refer Fig. \ref{Fig:Case-4}), a path of length 1 using nodes $v_1$ and $v_4$ is formed, and a path of length 1 using nodes $v_2$ and $v_3$ is formed. As per Algorithm \ref{algorithm_EXP1}, agents stay in their position in $r=4$. At round $r=5$ (refer Fig. \ref{Fig:Case-5}), node $v_1$ becomes an isolated node and a path $P_3=v_4\sim v_2\sim v_3$ is formed. As per Algorithm \ref{algorithm_EXP1} at round $r=5$, the agent moves to node $v_2$ from node $v_4$. If $r$(mod 6)=0, the adversary maintains $\mathcal{G}_r=\mathcal{G}_0$. If $r$(mod 6)=1, the adversary maintains $\mathcal{G}_r=\mathcal{G}_1$. If $r$(mod 6)=2, the adversary maintains $\mathcal{G}_r=\mathcal{G}_2$. If $r$(mod 6)=3, the adversary maintains $\mathcal{G}_r=\mathcal{G}_3$. If $r$(mod 6)=4, the adversary maintains $\mathcal{G}_r=\mathcal{G}_4$. If $r$(mod 6)=5, the adversary maintains $\mathcal{G}_r=\mathcal{G}_5$. Clearly, the adversary maintains the definition of $T(=6)$-Path Connectivity property. The dispersion is not achieved as per our algorithm, but the exploration has been achieved.
\end{remark}

\subsubsection{Correctness and Analysis of the Algorithm}
In this section, we prove that our algorithm solved the exploration problem. Let $T\geq 1$, and $l\geq 2$ be the number of holes be in $\mathcal{G}_r$, where $r$ (mod $T$)=0. 
\begin{lemma}\label{lm:exp-correc}
    If $l$ is not reduced by 1 between $t\in [r, \,r+T-1]$ and there is at least one multinode at the beginning of round $r$, where $r$ (mod $T$)=0, then the exploration has been achieved. 
\end{lemma}
\begin{proof}
    Let $v_1$, $v_2$,\ldots, $v_k$ be multinode in $\mathcal{G}_r$ at round $r$, $k\geq 1$, and $w_1$, $w_2$, \ldots, $w_{k'}$ be holes at round $r$, $r$ (mod $T$)=0 (some of these holes can be explored nodes). Let $G_1$ be a connected component of $\mathcal{G}_r$ with at least one hole. As per Algorithm \ref{algorithm_EXP1}, the number of holes can not be reduced in $G_1$, if there is no multinode in $G_1$.
    
    If $l$ is not reduced between round $r$ and $r+T-1$, then $v_1$ remains multinode between $r$ and $r+T-1$ rounds. As per the definition of $T$-Path Connectivity
    , there exists $t\in[r, \,r+T-1]$ such that there is a path between $v_1$ and $w_i$ for every $i \in [1, \, k']$. It implies that node $v_1$ and node $w_i$ are in the same connected component (say $G'$) at some round $t\in[r, r+T-1]$. At round $t$, if the number of holes is not reduced in $G'$, there is an agent at each node of $G'$. It implies an agent is at node $w_i$ at round $t$. This completes the proof. 
\end{proof}

\begin{theorem}\label{thm:exp-T-path-sol}
    Our algorithm solves exploration with $n-1$ agents in $\Theta(n\cdot T)$ rounds using $O(\log n)$ bits of memory per agent in the synchronous setting with global communication and 1-hop visibility.
\end{theorem}
\begin{proof}
    Due to Lemma \ref{lm:exp-correc}, we can say with the first $(n-1) \cdot T$ rounds, either the agents dispersed or the exploration is achieved. If agents achieve dispersion, then $n-1$ agents are at $n-1$ distinct nodes. In this case, there is one node (say $v$) which is a hole which might not be explored yet. Due to $T$-Path Connectivity within the next $T$ rounds, the node $v$ is connected to one node, which has one agent, and it gets explored by the agent as per our algorithm. The time complexity of the exploration is $O(n\cdot T)$, and the memory per agent is $O(\log n)$ bits. Due to Theorem \ref{thm:time_lower_exp}, our algorithm solves exploration with $n-1$ agents in $\Theta(n\cdot T)$ rounds using $O(\log n)$ bits of memory per agent in the synchronous setting with global communication and 1-hop visibility.
\end{proof}

\section{Conclusion}
In this work, we introduce a new connectivity model for dynamic graphs, namely, the $T$-Path Connectivity. We compare our model with existing connectivity models and discuss the status of the dispersion problem within those models, providing several impossibility results. Additionally, we study the exploration problem across all three connectivity models and provide impossibility results and optimal solutions in most cases. In the case of the Connectivity Time model, we show that it is impossible for \( k \leq n \) to solve the exploration problem. A future question to consider is: what is the value of \( k \) for which the exploration problem in Connectivity Time graphs becomes solvable?


\section{Acknowledgement}
Ashish Saxena would like to acknowledge the financial support from IIT Ropar. Kaushik Mondal would like to acknowledge the ISIRD grant provided by IIT Ropar.

\bibliography{bib}

\end{document}